\documentclass[11pt]{article}
\usepackage{tabularx} 
\usepackage{amsmath}  
\usepackage{graphicx} 
\usepackage{fullpage}
\usepackage{cite} 
\usepackage[final]{hyperref} 
\usepackage{delarray}
\usepackage{times}
\usepackage{subfigure}
\usepackage{color}
\usepackage{graphicx,amssymb,amsmath}
\usepackage{xspace}
\usepackage{enumerate}
\usepackage{algorithm}
\usepackage{algorithmic}
\usepackage{multirow}
\usepackage{nameref}
\usepackage{booktabs}
\usepackage{mdframed}

\newcommand{\commentout}[1]{}
\newcommand{\eat}[1]{}
\newcommand{\topic}[1]{\vspace{5pt}\noindent{{\bf #1:}}}

\newcommand{\SA}{\mathsf{SA}}

\newenvironment{proof}{\noindent {\em Proof: }\ignorespaces}%
                {\hspace*{\fill}$\Box$\par}
        {\hspace*{\fill}$\Box$\par\vspace{4mm}}
        {\hspace*{\fill}$\Box$\par}
\makeatletter
\def\namedlabel#1#2{\begingroup
   \def\@currentlabel{#2}%
   \label{#1}\endgroup
}
\makeatother

\newtheorem{theorem}{Theorem}
\newtheorem{lemma}{Lemma}

\newtheorem{definition}{Definition}
\newtheorem{property}{Property}
\newtheorem{observation}{Observation}

\hypersetup{
    colorlinks=true,       
    linkcolor=blue,    
    citecolor=blue,        
    filecolor=magenta,     
    urlcolor=blue
}

\newcommand{\suf}[1]{$\mathsf{suf}(#1)$}

\newcommand*{\QEDB}{\hfill\ensuremath{\square}}
\newenvironment{example}
{ \vspace{4pt}
	\begin{small}
	\noindent{\bf Example}:  }%
                {
                	\end{small}
                \vspace{4pt} }
\newcommand{\redb}[1]{\textcolor{red}{\mathbf{#1}}}
\newcommand{\udl}[1]{\underline{#1}}
\newcommand{\urb}[1]{\udl{\redb{#1}}}
\newcommand{\vrb}[1]{\overrightarrow{\redb{#1}}}
\newcommand{\vlrb}[1]{\overleftarrow{\redb{#1}}}
\newcommand{\LMS}{\mathsf{LMS}}
\newcommand{\sfindex}{\mathsf{Index}}
\newcommand{\type}{\mathsf{Type}}
\newcommand{\Suf}{\mathsf{suf}}
\newcommand{\Unique}{\mathsf{Unique}}
\newcommand{\Empty}{\mathsf{Empty}}
\newcommand{\NUnique}{\mathsf{Multi}}
\newcommand{\Uniquei}{\mathsf{Unique1}}
\newcommand{\NUniquei}{\mathsf{Multi1}}
\newcommand{\Uniqueii}{\mathsf{Unique2}}
\newcommand{\NUniqueii}{\mathsf{Multi2}}
\newcommand{\select}{\mathsf{select}}
\newcommand{\Ept}{\mathsf{E}}
\newcommand{\BH}{\mathsf{B_H}}
\newcommand{\BT}{\mathsf{B_T}}
\newcommand{\Rone}{\mathsf{R_1}}
\newcommand{\Rtwo}{\mathsf{R_2}}
\begin{document}

\title{Optimal In-Place Suffix Sorting}
\author{Zhize Li\\ IIIS, Tsinghua University \\{\small zz-li14@mails.tsinghua.edu.cn}
        \and Jian Li\\ IIIS, Tsinghua University \\ {\small lijian83@mail.tsinghua.edu.cn} 
        \and Hongwei Huo\\ SCST, Xidian University \\ {\small hwhuo@mail.xidian.edu.cn}
        }

\date{}
\maketitle


\begin{abstract}
The suffix array is a fundamental data structure for many applications that involve string searching and data compression. Designing time/space-efficient suffix array construction algorithms has attracted significant attention and considerable advances have been made for the past 20 years. We obtain the \emph{first} in-place suffix array construction algorithms that are optimal both in time and space for (read-only) integer alphabets. Concretely, we make the following contributions:
\begin{enumerate}
    \item
    For integer alphabets, we obtain the first suffix sorting algorithm which takes linear time and uses only $O(1)$ workspace (the workspace is the total space needed beyond the input string and the output suffix array). The input string may be modified during the execution of the algorithm, but should be restored upon termination of the algorithm.
    \item
    We strengthen the first result by providing the first in-place linear time algorithm for read-only integer alphabets with $|\Sigma|=O(n)$ (i.e., the input string cannot be modified). This algorithm settles the open problem posed by Franceschini and Muthukrishnan in ICALP 2007. The open problem asked to design in-place algorithms in $o(n\log n)$ time and ultimately, in $O(n)$ time for (read-only) integer alphabets with $|\Sigma| \leq n$. Our result is in fact slightly stronger since we allow $|\Sigma|=O(n)$.
    \item
    Besides, for the read-only general alphabets (i.e., only comparisons are allowed), we present an optimal in-place $O(n\log n)$ time suffix sorting algorithm, recovering the result obtained by Franceschini and Muthukrishnan which was an open problem posed by Manzini and Ferragina in ESA 2002.
\end{enumerate}
\end{abstract}



\section{Introduction}

In SODA 1990, suffix arrays were introduced by Manber and Myers~\cite{manber1990suffix}
as a space-saving alternative to suffix trees~\cite{McCreight1976,Farach1997}.
Since then, it has been used as a fundamental data structure for many applications in string processing, data compression, text indexing, information retrieval and computational biology~\cite{Ferragina2000,Abouelhoda2002,Grossi2005,huo2014practical,huo2016cs2a}.
Particularly, the suffix arrays are often used to compute the Burrows-Wheeler transform ~\cite{Burrows1994} and Lempel-Ziv factorization~\cite{Ziv1978}.
Comparing with suffix trees, suffix arrays use much less space in practice.
Abouelhoda et al.~\cite{Abouelhoda2004} showed that any problem which can be computed using suffix trees can also be solved using suffix arrays with the same asymptotic time complexity, which makes suffix arrays very attractive both in
theory and in practice.
Hence, suffix arrays have been studied extensively over the last 20 years (see e.g., \cite{karkkainen2003simple,ko2003space,Kaerkkaeinen2006,franceschini2007place,Nong2009,Nong2011,Nong2013}).
We refer the readers to the surveys~\cite{Puglisi2007,Dhaliwal2012} for many suffix sorting algorithms.

In 1990,  Manber and Myers~\cite{manber1990suffix} obtained the first $O(n\log n)$ time suffix sorting algorithm over general alphabets. In 2003, Ko and Aluru \cite{ko2003space}, K{\"a}rkk{\"a}inen and Sanders \cite{karkkainen2003simple} and Kim et al. \cite{Kim2003} independently obtained the first linear time algorithm for suffix sorting
over integer alphabets. Clearly, these algorithms are optimal in terms of asymptotic time complexity.
However, in many applications, the computational bottleneck is the \emph{space} as we need the space-saving suffix arrays instead of suffix trees, and significant efforts have been made in developing \emph{lightweight} (in terms of space usage) suffix sorting algorithms for the last decade (see e.g.,~\cite{Manzini2002,Burkhardt2003,ko2003space,Hon2003,Maniscalco2006,franceschini2007place,Nong2007,Nong2009,Nong2011,Nong2013}).
In particular, the ultimate goal in this line of work is to
obtain {\em in-place algorithms} (i.e., $O(1)$ additional space), which are also asymptotically optimal in time.

\vspace{-2.5mm}
\subsection{Problem Setting}
\vspace{-3mm}
\topic{Problem} Given a string $T=T[0\ldots n-1]$ with $n$ characters, we need to construct the \emph{suffix array} ($\SA$) which contains the \emph{indices} of all sorted suffixes of $T$ (see Definition \ref{def:sa} for the formal definition of $\SA$).

We consider the following three popular settings.
Note that the constant alphabets (e.g., ASCII code) is a special case of integer alphabets, and the (read-only) integer alphabets is commonly used in practice.
We measure the space usage of an algorithm in the unit of {\em words} same as \cite{franceschini2007place,Nong2013}.
A word contains $\lceil\log n\rceil$ bits.
One standard arithmetic or bitwise boolean operation on word-sized operands costs $O(1)$ time.

\begin{enumerate}
  \item Integer alphabets:  Each $T[i]\in[1,|\Sigma|]$ where the cardinality of the alphabets is $|\Sigma| \leq n$ and each $T[i]$ is stored in a word.
  The input string $T$ may be modified by the algorithm, but should be restored upon termination of the algorithm.

  \item Read-only integer alphabets:
  Each $T[i]\in[1,|\Sigma|]$ where $|\Sigma|=O(n)$.
  Moreover, the input string $T$ is \emph{read-only}.
  Each $T[i]$ can be read in $O(1)$ time.
  Note that we allow $|\Sigma|=O(n)$ rather than $|\Sigma| \leq n$ in Case 1.

  \item Read-only general alphabets:
    The only operations allowed on the characters of $T$ are comparisons.
    The input string $T$ is read-only and we assume that each comparison takes $O(1)$ time.
    We {\em cannot} write the input space, make bit operations, even
    copy an input character $T[i]$ to the work space.
    Clearly, $\Omega(n\log n)$ time is a lower bound for suffix sorting in this case,
    as it generalizes comparison-based sorting.
  \end{enumerate}

The \emph{workspace} used by an algorithm
is the total space needed by the algorithm,
excluding the space required by the input string $T$ and
the output suffix array $\SA$.
An algorithm which uses $O(1)$ words workspace to construct $\SA$
is called an \emph{in-place} algorithm.
See Tables~\ref{tab:int}
\footnote{Some previous algorithms state the space usages in terms of bits. We convert them into words.} and \ref{tab:gel} for an overview.

\vspace{-2.5mm}
\subsection{Related Work and Our Contributions}
\vspace{-1mm}
\subsubsection{Integer Alphabets}
\vspace{-1.5mm}
In this case, we allow the algorithm to modify the string $T$ during the execution of the algorithm.
We also describe how to restore $T$ (if needed) at the end of the algorithm in Appendix \ref{app:rec}.
Chan et al.~\cite{chan2014selection} denote this model as the \emph{restore model} in their paper.
We list several previous results and our new result in Table~\ref{tab:int}.
Earlier algorithms that require more than $O(n)$ words workspace
(see Table~\ref{tab:int})
do not need to modify the string $T$ as they can afford to create a new array with $n$ words
to store the input.

Nong et al.~\cite{Nong2009a,Nong2011} obtained the first nearly linear time
algorithm that used sublinear workspace.
Recently, Nong \cite{Nong2013} obtained a linear time algorithm which used $|\Sigma|$ words workspace
without modifying the string $T$.
We improve their results as in the following theorem.
\vspace{-1.5mm}
\begin{theorem}
\label{thm:mainthm1}
There is an in-place linear time algorithm for suffix sorting over integer alphabets with $|\Sigma| \leq n$.
\end{theorem}

Our algorithm is based on the induced sorting framework developed in \cite{ko2003space} (which is also used in several previous algorithms \cite{franceschini2007place,Puglisi2007,Nong2009,Nong2009a,Nong2011,Nong2013}).
We develop a few elementary, yet effective tricks to further reduce the space usage to constant. The proposed algorithm and the new tricks are also useful
for the read-only integer and general alphabets.

\vspace{-2mm}
\subsubsection{Read-only Integer Alphabets}
\label{sec1:roint}
\vspace{-1mm}
\begin{table}[t]
  \caption{Time and workspace of suffix sorting algorithms for (read-only) integer alphabets $\Sigma$}
  \vspace{1mm}
 \label{tab:int}
 \centering
 \begin{tabular}{ccc}
  \toprule
Time & Workspace (words) & Algorithms\\
  \midrule
        $O(n^2)$ & $O(n)$ & ~\cite{Schuermann2007} \\
	$O(n\log^2 n)$ & $O(n)$ & ~\cite{Sadakane1998} \\
  $O(n\sqrt{|\Sigma|\log (n/|\Sigma|)})$ & $O(n)$ &                                          ~\cite{Baron2005}\\
  $O(n\log n)$   & $O(n)$ & ~\cite{manber1990suffix,Larsson2007} \\
  $O(n\log\log n)$ & $O(n)$ & ~\cite{Kim2004} \\
  $O(n)$ & $O(n)$ &  ~\cite{Kim2003,karkkainen2003simple,ko2003space}\\
        $O(n\log\log |\Sigma|)$ & $O(n\log |\Sigma|/\log n)$ & ~\cite{Hon2003} \\
        $O(vn)$ & $O(n/\sqrt{v})~~v\in[1,\sqrt{n}]$ & ~\cite{Kaerkkaeinen2006} \\
        $O(n)$ & $n + n/\log n +O(1)$ & ~\cite{Nong2009} \\
        $O(n^2\log n)$ & $cn+O(1)~~ c<1$ & ~\cite{Manzini2002,Maniscalco2006} \\
        $O(\frac{1}{\epsilon}n)$ $^\star$ & $n^\epsilon +n/\log n +O(1)$
	        \footnotemark
          & ~\cite{Nong2009a,Nong2011} \\
        $O(n^2\log n)$ & $|\Sigma|+O(1)$ & ~\cite{Itoh1999} \\
		$O(n\log |\Sigma|)$ & $|\Sigma| +O(1)$ & ~\cite{Nong2007} \\
        $O(n)$ & $|\Sigma| + O(1)$ & ~\cite{Nong2013} \\
        $O(n)$ & $O(1)$ & This paper \\
  \bottomrule
 \end{tabular}
\vspace{1mm}

$T$ is read-only in all algorithms except in the fifth to last row (marked with $^\star$).
\end{table}
\footnotetext{
	Nong et al. \cite{Nong2009a,Nong2011} assumed that the word size is 32 bits and any integer can fit into one word.
	 The result listed here is under the standard assumption that a word contains $\lceil\log n\rceil$ bits.
	 It is not hard to verify that the bucket array $B$ in their algorithm requires $n^\epsilon$ words. They also need an $n$ bits array (or equivalently $n/\log n$ words).
   }

Now, we consider the more difficult case where the input
string $T$ is read-only.
This is the main contribution of this paper.
There are many existing algorithms for this case.
See Table \ref{tab:int} for an overview.
In ICALP 2007, Franceschini and Muthukrishnan~\cite{franceschini2007place}
posed an open problem for designing an in-place algorithm that takes $o(n\log n)$ time or ultimately $O(n)$ time for (read-only) integer alphabets with $|\Sigma| \leq n$ (in fact, they did not specify whether the input
string $T$ is read-only or not).
The current best result along this line is provided by
Nong~\cite{Nong2013}, which used $|\Sigma|$ words workspace (Nong's algorithm is in-place if $|\Sigma|=O(1)$,
i.e., constant alphabets).
Note that in the worst case $|\Sigma|$ can be as large as $O(n)$.

In this paper, we settle down this open problem by providing the first optimal linear time in-place algorithm,
as in the following theorem.
Note that our result is in fact slightly stronger since we allow $|\Sigma|=O(n)$ instead of $|\Sigma| \leq n$ mentioned in the open problem~\cite{franceschini2007place}.
\vspace{-1mm}
\begin{theorem}[Main Theorem]
	\label{thm:mainthm2}
	There is an in-place linear time algorithm for suffix sorting over integer alphabets, even if the input string $T$ is read-only and the size of the alphabet $|\Sigma|$ is $O(n)$.
\vspace{2mm}
\end{theorem}

\subsubsection{Read-only General Alphabets}
Now, we consider the case where the only operations allowed on the characters of string $T$ (read-only) are comparisons.
See Table \ref{tab:gel} for an overview of the results.
In 2002, Manzini and Ferragina~\cite{Manzini2002} posed an open problem, which asked whether there exists an $O(n\log n)$ time algorithm using $o(n)$ workspace. In 2007, Franceschini and Muthukrishnan~\cite{franceschini2007place} obtained
the first in-place algorithm that runs in optimal $O(n\log n)$ time.
Their conference paper is somewhat complicated and densely-argued.

We also give an optimal in-place algorithm which achieves
the same result, as in the following theorem.
In addition, our algorithm does not make any bit operations while theirs uses bit operations heavily.
Our algorithm is also arguably simpler.

\begin{theorem}
	\label{thm:mainthm3}
	There is an in-place $O(n\log n)$ time algorithm for suffix sorting over general alphabets, even if the input string $T$ is read-only and only comparisons between characters are allowed.
\end{theorem}

\begin{table}[t]
  \caption{Time and workspace of suffix sorting algorithms for read-only general alphabets}
  \vspace{1mm}
 \label{tab:gel}
 \centering
 \begin{tabular}{ccc}
  \toprule
Time & Workspace(words) & Algorithms\\
  \midrule
  $O(n\log n)$   & $O(n)$ & ~\cite{manber1990suffix,Larsson2007} \\
  $O(vn+n\log n)$ & $O(v+n/\sqrt{v})~~v\in[2,n]$ & ~\cite{Burkhardt2003} \\
  $O(vn+n\log n)$ & $O(n/\sqrt{v})~~v\in[1,\sqrt{n}]$ & ~\cite{Kaerkkaeinen2006} \\
  $O(n\log n)$ & $O(1)$ & ~\cite{franceschini2007place}
\\
  $O(n\log n)$ & $O(1)$ & This paper \\
  \bottomrule
 \end{tabular}
\end{table}

\vspace{-2mm}
\subsection{Difficulties and Our Approach}
\label{sec:diff}
\subsubsection{Difficulties}
Typically, the suffix sorting algorithms are recursive algorithms. The size of the recursive (reduced) sub-problem is usually less than half of the current problem. See  e.g.,~\cite{ko2003space,Kaerkkaeinen2006,Puglisi2007,franceschini2007place,Nong2009,Nong2011,Nong2013}.
However, all previous algorithms require extra arrays, e.g., \emph{bucket array} (which needs $|\Sigma|$ words at the top recursive level and $n/2$ words at the deep recursive levels), \emph{type array} (which needs $n/\log n$ words) and/or other auxiliary arrays (which need up to $O(n)$ words), to construct the reduced problems and use the results of the reduced problems to sort the original suffixes.
\footnote{
The definitions of bucket array and type array can be found in Section~\ref{sec:pre}.
}

In particular, Nong et al.~\cite{Nong2009} made a breakthrough by providing the SA-IS algorithm which only required one bucket array (which needs $\max\{|\Sigma|,n/2\}$ words) and one type array ($n/\log n$ words).
Note that the bucket array and type array are reused for each recursive level.

Currently, the best result is provided by Nong~\cite{Nong2013}. However, Nong's algorithm  requires the bucket array at the top recursive level, but not required at the deep levels.
Hence, it needs $|\Sigma|$ words instead of $\max\{|\Sigma|,n/2\}$ words.
Note that $|\Sigma|$ can be $O(n)$ in the worst case for integer alphabets.
For the type array, Nong used this bucket array to implicitly indicate the type information.

Thus, \emph{the main technical difficulty is to remove the workspace for the bucket array at the top recursive level since there is no extra space to use}.
Note that it is non-trivial since $T$ is read-only and $\SA$ needs to store the final order of all suffixes.
Besides, the previous sorting steps or tricks may not work if one removes the bucket array.
For example, Nong~\cite{Nong2013} used the bucket array to indicate the type information.
If the bucket array is removed, then one would need the type array.

\subsubsection{Our Approach}
\vspace{-0.5mm}
We briefly describe our optimal in-place linear time suffix sorting algorithms that overcome these difficulties.
We provide an \emph{interior counter trick} which can implicitly represent the dynamic \emph{LF/RF-entry} information
(see Section~\ref{sec:pre} for the definition)
in $\SA$. Besides, we provide a \emph{pointer data structure} which can represent the bucket heads/tails in $\SA$.
Combining these two techniques, we can remove the workspace needed by the bucket array entirely.
Note that it is non-trivial for the top recursive level which is the most difficult part, since the pointer data structure needs \emph{nonconstant} workspace and we only have $O(1)$ extra workspace.
In our algorithm, we divide the sorting step into two stages to address this issue.
In order to remove the type array, we provide some useful properties and observations
which allows us to retrieve the type information efficiently.
For the general alphabets case, we provide simple sorting steps and extend the interior counter trick to obtain
an optimal in-place $O(n\log n)$ time suffix sorting algorithm.


\topic{Organization}
The remaining of the paper is organized as follows.
Section~\ref{sec:pre} covers the preliminary knowledge. In Section~\ref{sec:int}, \ref{sec:roint} and \ref{sec:gel}, we describe the framework and the details of our optimal in-place suffix sorting algorithms for the integer alphabets, read-only integer alphabets and read-only general alphabets, respectively.
Finally, we conclude in Section~\ref{sec:con}.

\vspace{-1mm}
\section{Preliminaries}
\label{sec:pre}
\vspace{-1mm}
Given a string $T=T[0\ldots n-1]$ with $n$ characters, the suffixes of $T$ are $T[i\ldots n-1]$ for all $i\in[0,n-1]$,  where $T[i\ldots j]$ denotes the substring $T[i]T[i+1]\ldots T[j]$ in $T$.
To simplify the argument, we assume that the final character $T[n-1]$ is a sentinel
which is lexicographically smaller than any other characters in $\Sigma$.
Without loss of generality, we assume that $T[n-1]=0$.
\footnote{
   Some previous papers use $\$$ to denote the sentinel. We use $0$ here since we consider the integer alphabets.
}
Any two suffixes in $T$ must be different since their lengths are different,
and their lexicographical order can be determined by comparing their characters one by one until we see a difference due to the existence of the sentinel.

\begin{definition}
\label{def:sa}
The suffix array $\SA$ contains the indices of all suffixes of $T$ which are sorted in lexicographical order, i.e., $\Suf(\SA[i])< \Suf(\SA[j])$ for all $i<j$, where \suf{i} denotes the suffix $T[i\ldots n-1]$.
\end{definition}
For example, if $T=``1220"$, then all suffixes are $\{1220,220,20,0\}$ and $\SA=[3,0,2,1]$. Note that $\SA$ always uses $n$ words no matter what the alphabets $\Sigma$ are, since it contains the permutation of $\{0,\ldots,n-1\}$, where $n$ is the length of $T$.

A suffix \suf{i} is said to be \emph{S-suffix} (S-type suffix) if \suf{i} $<$ \suf{i+1}. Otherwise, it is \emph{L-suffix} (L-type suffix)~\cite{ko2003space}.
The last suffix \suf{n-1} containing only the single character $0$ (the sentinel) is defined to be an S-suffix.
Equivalently,
the \suf{i} is S-suffix if and only if (1) $i=n-1$; or (2) $T[i]<T[i+1]$; or (3) $T[i]=T[i+1]$ and \suf{i+1} is S-suffix.
Obviously, the types can be computed by a linear scan of $T$
(from $T[n-1]$ to $T[0]$).
We further define the type of a character $T[i]$ to \emph{S-type} (or \emph{L-type} resp.) if \suf{i} is S-suffix (or L-suffix resp.).
A substring $T[i\ldots j]$ is called an {\em S-substring} if both $T[i]$ and $T[j]$ are S-type, and there is no other S-type characters between them, or $i=j=n-1$ (the single sentinel).
We can define \emph{L-substring} similarly.

\begin{example}
We use the following running example for the integer alphabets case throughout the paper. Consider a string $T[0\ldots 12]=``2113311331210"$.
\[
\begin{array}{cccccccccccccc}
\mathsf{Index} & 0 & 1 & 2 & 3 & 4 & 5 & 6 & 7 & 8 & 9 & 10& 11& 12 \\
  T    & 2 & 1 & 1 & 3 & 3 & 1 & 1 & 3 & 3 & 1 & 2 & 1 & 0 \\
\mathsf{Type}   & L & S & S & L & L & S & S & L & L & S & L & L & S
\end{array}
\]
E.g., $T[2]$ is S-type since $T[2]=1<T[3]=3$. The S-substrings are $\{11,1331,11,1331,1210,0\}$.
\QEDB
\end{example}

Obviously, the indices of all suffixes, which begin with the same character, must appear consecutively  in $\SA$. We denote a subarray in $\SA$ for these suffixes with the same beginning character as a {\em bucket}, where the \emph{head} and the \emph{tail} of a bucket refer to the first and the last index of the bucket in $\SA$ respectively.
Moreover, we define the first common character as its {\em bucket character}.
We often use the bucket character to index the bucket. For example, if the bucket character
is $T[i]$, we refer to this bucket as bucket $T[i]$.
Sometimes we say that we place suffix \suf{i} of $T$ into $\SA$, it always means that we place its corresponding index $i$ into $\SA$ since \suf{i} is a substring.

The {\em induced sorting} technique, developed by Ko and Aluru \cite{ko2003space}, is responsible for many
recent advances of suffix sorting algorithms~\cite{Puglisi2007,franceschini2007place, Nong2009,Nong2009a,Nong2011,Nong2013}, and is also crucial to us.
It can be used to
induce the lexicographical order of L-suffixes from
the sorted S-suffixes. Before introducing the induce sorting technique, we need the following useful property with respect to L-suffixes and S-suffixes (the proof simply follows from the definition of L- and S-suffix).
\begin{property}\cite{ko2003space}\label{prop:ls}
In any bucket, S-suffixes always appear after the L-suffixes in $\SA$, i.e.,  if an S-suffix and an L-suffix begin with the same character, the  L-suffix is always smaller than the S-suffix.
\end{property}

Now, we briefly introduce the standard induced sorting technique which needs the \emph{bucket array} and \emph{type array}
explicitly.
The bucket array contains $|\Sigma|$ integers and each denotes the position of a bucket head/tail in $\SA$.
The type array contains $n$ bits and each entry denotes an L/S-type information for $T$ (i.e., $0/1$ for L/S-type).

\topic{Inducing the order of L-suffixes
	from the sorted S-suffixes}
Assume that all indices of the sorted S-suffixes are already in their correct positions in
$\SA$ (i.e., in the tail of their corresponding buckets in $\SA$).
Now, we define some new notations (e.g., \emph{LF/RF-entry}) to simplify the representation.
We scan $\SA$ from left to right (i.e., from $\SA[0]$ to $\SA[n-1]$).
We maintain an \emph{LF-pointer} (leftmost free pointer) for each bucket which points to the leftmost free entry
(called the \emph{LF-entry}) of the bucket.
The LF-pointers initially point to the head of their corresponding buckets.
When we scan $\SA[i]$, let $j=\SA[i]-1$.
If $\Suf(j)$ is an L-suffix (indicated by the type array),
we place the index of $\Suf(j)$ (i.e., $j$) into the LF-entry of bucket $T[j]$,
and then let the LF-pointer of this bucket $T[j]$ point to the next free entry.
The LF-pointers are maintained in the bucket array.
If $\Suf(j)$ is an S-suffix, we do nothing (since all S-suffixes are already sorted in the correct positions). We give a running example in Appendix~\ref{app:inducedsort}.

Sorting all S-suffixes from the sorted L-suffixes is completely symmetrical:
we scan $\SA$ from right to left, maintaining an \emph{RF-pointer} (rightmost free pointer) for each bucket which points to the \emph{RF-entry} (rightmost free entry) of the bucket.

\begin{lemma}~\cite{ko2003space}\label{lem:ls}
Suppose all S-suffixes (or L-suffixes resp.) of $T$ are already sorted.
Then using induced sorting, all L-suffixes (or S-suffixes resp.) can be sorted correctly.
\end{lemma}

The idea of induced sorting is that the lexicographical order between $\Suf(i)$ and $\Suf(j)$ is decided by the order of $\Suf(i+1)$ and $\Suf(j+1)$ if $\Suf(i)$ and $\Suf(j)$ are in the same bucket (i.e., $T[i]=T[j]$). We only need to specify the correct order of these L-suffixes in the same buckets since we always place the L-suffixes in their corresponding buckets. Consider two L-suffixes $\Suf(i)$ and $\Suf(j)$ in the same bucket. We have $\Suf(i+1)<\Suf(i)$ and $\Suf(j+1)<\Suf(j)$ by the definition of L-suffix. Since we scan $\SA$ from left to right, $\Suf(i+1)$ and $\Suf(j+1)$ must appear earlier than $\Suf(i)$ and $\Suf(j)$. Hence the correctness of induced sorting is not hard to prove by induction.

\topic{Inducing the order of L-suffixes from the sorted LMS-suffixes}
A suffix \suf{i} is called an {\em LMS-suffix} (leftmost S-type) if $T[i]$ is S-type and $T[i-1]$ is L-type, for $i\geq 1$.
Nong et al.~\cite{Nong2009} observed that we can sort all L-suffixes from the sorted LMS-suffixes (instead of S-suffixes) if they are stored in the tail of their corresponding buckets in $\SA$. Roughly speaking, the idea is that in the induced sorting, only LMS-suffixes are useful for sorting L-suffixes.
One difference from the standard induced sorting is that we may scan some empty entries in $\SA$. However, the empty entries can be ignored and all L-suffixes can still be sorted correctly.
We provide a running example in Appendix~\ref{app:lfromlms}.

\begin{lemma}~\cite{Nong2009}\label{lem:lmsl}
Suppose all LMS-suffixes of $T$ are already sorted and stored in the tail of their buckets. Then using induced sorting, all L-suffixes can be sorted correctly.
\end{lemma}

After we sort all L-suffixes from the sorted LMS-suffixes,
we can induce the order of all S-suffixes from the sorted L-suffixes by Lemma \ref{lem:ls}, and sort all suffixes.
Now, we introduce how to sort the LMS-suffixes.
First, we define some notations. A character $T[i]$ of $T$ is called \emph{LMS-character} if \suf{i} is LMS-suffix. A substring $T[i\ldots j]$ is called an {\em LMS-substring} if both $T[i]$ and $T[j]$ are LMS-characters, and there is no other LMS-characters between them, or $i=j=n-1$ (the single sentinel). Similarly, we can define \emph{LML-suffix} (leftmost L-type) and \emph{LML-substring}.

\topic{Sort the LMS-suffixes}
If we know the lexicographical order of all LMS-substrings, then we can use their ranks to construct the reduced problem $T_1$. Sorting the suffixes of $T_1$
is equivalent to sorting the LMS-suffixes of $T$ (see the following example).
Nong et al.~\cite{Nong2009} showed that we can use the same induced sorting step to sort all LMS-substrings from the sorted LMS-characters of $T$. We briefly sketch their idea. We refer the readers to ~\cite{Nong2009} for the details.
We define the \emph{LMS-prefix} of a suffix $\Suf(i)$ to be $T[i\ldots j]$, where $j>i$ is the smallest position in $\Suf(i)$ such that $T[j]$ is an LMS character (e.g., the LMS-prefix of $\Suf(4)$ is $``31"$).
Suppose all LMS-characters are stored in the tail of their corresponding buckets in $\SA$.
First, we sort all LMS-prefix of L-suffixes from the sorted LMS-characters,
using one scan of induced sorting from left to right (the same as induce the order of L-suffixes from LMS-suffixes).
Then we sort all LMS-prefix of S-suffixes from the sorted LMS-prefix of L-suffixes (in the same way as inducing
the order of S-suffixes from L-suffixes). After this, we have sorted all LMS-substrings since all LMS-substrings are LMS-prefix of S-suffixes by the definition of LMS-prefix. The correctness proof follows the same argument as in the standard setting.

\begin{example}
Continue the running example:
\[
\begin{array}{cccccccccccccc}
\mathsf{Index}  & 0 & 1 & 2 & 3 & 4 & 5 & 6 & 7 & 8 & 9 & 10& 11& 12 \\
  T    & 2 & 1 & 1 & 3 & 3 & 1 & 1 & 3 & 3 & 1 & 2 & 1 & 0 \\
\mathsf{Type}   & L & S & S & L & L & S & S & L & L & S & L & L & S \\
  \LMS  &   & * &   &   &   & * &   &   &   & * &   &   & *
\end{array}
\]
Note that the LMS-substrings are $\{11331,11331,1210,0\}$. Their ranks in lexicographical order are $\{1,1,2,0\}$. Thus, the reduced problem is $T_1=1120$.
The order of the suffixes of $T_1$
is the same as the order of corresponding LMS-suffixes of $T$.
The suffix array of the reduced problem $T_1$ is solved recursively.
\QEDB
\end{example}

Note that in this preliminary section, the induced sorting steps are not \emph{in-place} since they require explicit storage for the bucket and type arrays.

\section{Suffix Sorting for Integer Alphabets}
\label{sec:int}

\subsection{Framework}
To avoid confusion, we recall that an LMS-character is a single character, an LMS-substring is a substring which begins with an LMS-character and ends with an LMS-character,
and an LMS-suffix is a suffix of $T$ which begins with an LMS-character.
Our optimal in-place suffix sorting algorithm for integer alphabets
consists of the following steps.

\begin{enumerate}
  \item (Section~\ref{secint:renamet}) Rename $T$.
  \item (Section~\ref{secint:putlms}) Sort all LMS-characters of $T$.
  \item (Section~\ref{secint:sortlms}) Induced sort all LMS-substrings from the sorted LMS-characters.
  \item (Section~\ref{secint:gett1}) Construct the reduced problem $T_1$ (in which we need to sort all LMS-suffixes) from the sorted LMS-substrings.
  \item (Section~\ref{secint:sorttlms}) Sort the LMS-suffixes by solving $T_1$ recursively.
  \item (Section~\ref{secint:sortt}) Induced sort all suffixes of $T$ from the sorted LMS-suffixes.
\end{enumerate}

In a high level, the framework is similar to
several other previous algorithms based on induced sorting~\cite{ko2003space,franceschini2007place,Puglisi2007,Nong2009,Nong2009a,Nong2011,Nong2013}, and in particular to \cite{Nong2009}. Our algorithm differs in the detailed implementation of the above steps to obtain the first in-place algorithm (i.e., remove the bucket array and type array). We describe the details of the above steps
in the following sections.
We also describe how to restore $T$ (if needed) at the end of the algorithm in Appendix \ref{app:rec}.

\eat{
\begin{algorithm}[h]
\caption{Optimal in-place suffix sorting algorithm for integer alphabets}
\label{alg:1}
\begin{algorithmic}[1]
\REQUIRE $T$;
\STATE Rename $T$ to indicate the bucket heads and tails in $\SA$, using $\SA[0\ldots n-1]$ as the temporary space;
\STATE Sort all LMS-characters of $T$ and store them in the tail of their corresponding buckets in $\SA$, using the renamed $T$ and our new interior counter trick;
\STATE Induced sort all LMS-substrings from the sorted LMS-characters,  using the renamed $T$ and our interior counter trick;
\STATE Construct the reduced recursive problem $T_1$ using the ranks of the sorted LMS-substrings and store $T_1$ in $\SA[0\ldots n_1-1]$;
\STATE Let $\SA_1$ denote $\SA[n-n_1\ldots n-1]$ which is reused as the output space for problem $T_1$;
\IF{all characters in $T_1$ are unique}
\STATE Directly construct $\SA_1$ for $T_1$;
\ELSE \STATE Solve $T_1$ recursively and obtain the suffix array $\SA_1$;
\ENDIF
\STATE Sort all LMS-suffixes of $T$ from $\SA_1$ and store them in $\SA[n-n_1\ldots n-1]$, simply replacing $\SA_1$ with the indices of LMS-suffixes;
\STATE Induced sort all suffixes of $T$ from the sorted LMS-suffixes, using the renamed $T$ and our interior counter trick (this step is the same as Line 3);
\STATE Restore $T$ (if needed);
\RETURN $\SA$
\end{algorithmic}
\end{algorithm}
}

\subsection{Rename $T$}
\label{secint:renamet}
In this section, we rename each L-type character of $T$ to be the index of its bucket head and each S-type character of $T$ to be the index of its bucket tail (Nong et al.\cite{Nong2011} has a similar renaming step).
The correctness of the step is shown in the following Lemma~\ref{lem:rename}.

\begin{lemma}\label{lem:rename}
	The renaming step does not change the lexicographical order of all suffixes of $T$.
\end{lemma}
\begin{proof}
For any two suffixes, beginning with the same character,
the L-suffix is smaller than the S-suffix (Property \ref{prop:ls}).
Hence, the renaming step
does not change the relative orders of all suffixes.
\end{proof}

\begin{example}
We illustrate the renaming process in our running example.
\[
\begin{array}{cccccccccccccc}
\sfindex  & 0 & 1 & 2 & 3 & 4 & 5 & 6 & 7 & 8 & 9 & 10& 11& 12 \\
  T    & 2 & 1 & 1 & 3 & 3 & 1 & 1 & 3 & 3 & 1 & 2 & 1 & 0 \\
\type   & L & S & S & L & L & S & S & L & L & S & L & L & S \\
  \SA   &(12)&(11& 1 & 5 & 9 & 2& 6)&(10& 0)&(4 & 8 & 3 & 7) \\
\mathsf{Bucket} &(0)&(1 & 1 & 1 & 1 & 1 & 1)&(2 & 2)&(3 & 3 & 3 & 3)
\end{array}
\]
After renaming, we get $T'$ as following:
\[
\begin{array}{cccccccccccccc}
\sfindex  & 0 & 1 & 2 & 3 & 4 & 5 & 6 & 7 & 8 & 9 & 10& 11& 12 \\
  T'   & 7 & 6 & 6 & 9 & 9 & 6 & 6 & 9 & 9 & 6 & 7 & 1 & 0
\end{array}
\]
E.g., $T'[0]=7$ since $T[0]$ is L-type and the head of bucket $2$ (i.e., bucket $T[0]$) is $7$,
and $T'[1]=6$ since $T[1]$ is S-type and the tail of bucket $1$ (i.e., bucket $T[1]$) is $6$.
Note that the heads of bucket $0,1,2,3$ are $0,1,7,9,$ respectively.
\QEDB
\end{example}

Now, we describe how to implement this step using linear time and $O(1)$ workspace.
This step is very simple and similar to the counting sort (see e.g., \cite[Ch. 8]{thomas2001introduction}).
We first rename all L-type characters to be the index of its bucket head, and then
rename all S-type characters to be the index of its bucket tail.
\begin{enumerate}
  \item First we scan $T$ once to compute the number of times each character occurs in $T$ and store them in $\SA$.
      Then we perform a \emph{prefix sum computation} to determine the starting position of each character (i.e., bucket head) in $\SA$.
      Finally we scan $T$ once again to rename each character as the index of its bucket head.
  \item Now we need to let the S-type characters of $T$ to be the index of its bucket tail.
        Same as before, we scan $T$ to compute the number of times each character occurs in $T$ and store them in $\SA$.  Then, we scan $T$ once again from right to left. For each S-type $T[i]$, we let it be the index of its bucket tail (i.e., just add the number of characters belonging to this bucket to $T[i]$). Note that if we scan $T$ from right to left, for each $T[i]$, we can know its type is L-type or S-type in $O(1)$ time. There are two cases: 1) if $T[i]\neq T[i+1]$, we can know its type immediately by definition; 2) if $T[i]=T[i+1]$ then its type is the same as the type of $T[i+1]$. We only need to maintain \emph{one} boolean variable which represents the type of previous scanned character $T[i+1]$.
\end{enumerate}

\subsection{Sort all LMS-characters}
\label{secint:putlms}

Now, we sort all LMS-characters of $T$, i.e., place the indices of the LMS-characters in
the tail of their corresponding buckets in $\SA$.
Note that we do not have extra space to store
the LF/RF-pointers/counters for each bucket
to indicate the position of the free entries in the process.
For this purpose, we develop a simple trick,
called {\em interior counter trick},
which allows us to carefully use the space in $\SA$
to store the information of both the indices and the pointers.
The implementation details are described below.
In the steps, we use three special symbols which are $\Unique$, $\Empty$ and $\NUnique$.
\footnote{\label{footnote:special}
    We use at most five special symbols in this paper.
    The special symbol is only used to simplify the argument and we do not have to
    impose any additional assumption to accommodate these symbols (including the read-only cases in Section~\ref{sec:roint} and \ref{sec:gel}).
    These special symbols can be handled using an extra $O(1)$ workspace.
    We defer the details to Appendix \ref{app:special}.}
    \eat{
     Now, we explain the reasons. We distinguish two cases.

     1) If $n< 2^{\lceil\log n\rceil} -5$, we can use the largest five integers as the special symbols,
     and it is easy to see that there is no conflict with the indices of the suffixes.

     2) Otherwise, we use any five integers (say the largest five) in $(n/2,n-1]$ as the special symbols.
     Then we use ten extra variables (as the bucket array) to indicate the head/tail of these five buckets (which the five integers belong to) and their LF/RF-entries. We can obtain these five bucket heads/tails by scanning $T$ once to count how many characters are smaller/larger than these five characters, respectively.
     The reason why we need these five integers larger than $n/2$ is that there is at most one bucket in $\SA$ whose size can be larger than $n/2$.
     Since in our interior counter trick, we use a counter for each bucket to denote how many suffixes have been placed into this bucket. Thus, there is at most one counter in $\SA$ which may conflict with these five integers. In this case, we use two extra variables to indicate this bucket information as we did for the five integers (note that we can identify whether this case exists or not by scanning the string $T$ once).
	}

\topic{Step 1. Initializing $\SA$}
First we clear $\SA$ (i.e., $\SA[i]= \Empty$, for all $i\in [0,n-1]$). Then we scan $T$ from right to left. For every $T[i]$ which is an LMS-character (this can be easily decided in constant time), do the following:
        \begin{enumerate}[(1)]
           \item If $\SA[T[i]]= \Empty$, let $\SA[T[i]]= \Unique$ (meaning it is the unique LMS-character in this bucket).
           Note that after the renaming step, $T[i]$ is the
           index of its bucket tail.
           \item If $\SA[T[i]] = \Unique$, let $\SA[T[i]] = \NUnique$ (meaning the number of LMS-characters in this bucket is at least $2$).
            \item Otherwise, do nothing.
        \end{enumerate}

\topic{Step 2. Placing all indices of LMS-characters into $\SA$}
We scan $T$ from right to left. For every $T[i]$ which is an LMS-character, we distinguish the following cases:
\begin{enumerate}[(1)]
        \item $\SA[T[i]]=\Unique$:
        In this case, we let $\SA[T[i]] = i$ (i.e., $T[i]$ is the unique LMS-character in its bucket, and we just put its index into its bucket).
        \item $\SA[T[i]] = \NUnique$ and $\SA[T[i]-1]= \Empty$:
        In this case, $T[i]$ is the first (i.e. largest index, since we scan $T$ from right to left) LMS-character in its bucket.
        So if $\SA[T[i]-2]=\Empty$, we let $\SA[T[i]-2]=i$ and $\SA[T[i]-1] =1$ (i.e.,
        we use $\SA[T[i]-1]$ as the counter for the number of LMS-characters which has been added to this bucket so far).
        Otherwise, $\SA[T[i]-2]\neq \Empty$ (i.e., $\SA[T[i]-2]$ is in a different bucket, which implies that this bucket has only two LMS-characters).
        Then we let $\SA[T[i]]=i$ and $\SA[T[i]-1]$ be $\Empty$ (We do not need a counter in this case and the last LMS-character belonging to this bucket will be dealt with in the later process).
        \item $\SA[T[i]] = \NUnique$ and $\SA[T[i]-1]\neq \Empty$: In this case, $\SA[T[i]-1]$ is maintained as the counter.
        Let $c= \SA[T[i]-1]$. We check whether the position ($\SA[T[i]-c-2]$), i.e. $c+2$ positions before its tail, is $\Empty$ or not. If $\SA[T[i]-c-2]=\Empty$, let $\SA[T[i]-c-2]=i$ and increase $\SA[T[i]-1]$ by one (i.e., update the counter number).
        Otherwise $\SA[T[i]-c-2]\neq \Empty$ (i.e., reaching another bucket), we need to shift these $c$ indices to the right by two positions (i.e., move $\SA[T[i]-c-1\ldots T[i]-2]$ to $\SA[T[i]-c+1\ldots T[i]]$), and let $\SA[T[i]-c]= i$ and $\SA[T[i]-c-1] = \Empty$. After this, only one LMS-character needs to be added into this bucket in the later process.
        \item $\SA[T[i]]$ is an index:
        From case (2) and (3), we know the current $T[i]$ must be the last LMS-character in its bucket.
        So we scan $\SA$ from right to left, starting with  $\SA[T[i]]$, to find the first position $j$ such that $\SA[j] = \Empty$.
        Then we let $\SA[j] = i$. Now, we have filled the entire bucket. However, we note that not every
        bucket is fully filled as we have only processed LMS-characters so far.
\end{enumerate}
After the above Step 1 and 2, there may be still some special symbols $\NUnique$ and the counters
(because the bucket is not fully filled, so we have not shifted these indices to the right in the bucket).
We need to free these position. We scan $\SA$ once more from right to left.
If $\SA[i] = \NUnique$, we shift the indices of LMS-characters in this bucket to the right by two positions (i.e., $\SA[i-c-1\ldots i-2]$ to $\SA[i-c+1\ldots i]$)
and let $\SA[i-c-1] = \SA[i-c]= \Empty$, where $c=\SA[i-1]$ denotes the counter.

\begin{lemma}
The indices of the LMS-characters can be placed in the tail of their corresponding buckets in $\SA$ using linear time and $O(1)$ workspace.
\end{lemma}
\begin{proof}
We only need to show that the Step 2, i.e., placing all indices of LMS-characters into $\SA$, takes $O(n)$ time. For each scanned $T[i]$, it takes $O(1)$ time except when the $T[i]$ is the last two LMS-characters of its bucket. In this case, we need to shift the indices in this bucket when dealing with the penultimate LMS-character in its bucket, and scan the bucket when dealing with the last one. It takes $O(n)$ time since every bucket only needs to be shifted and scanned once. The space usage of this step is obvious.
\end{proof}

\begin{example}
Continue our example ($U,E$ and $M$ denote $\Unique$, $\Empty$ and $\NUnique$, respectively):

Step 1. Initializing $\SA$:
\[
\begin{array}{cccccccccccccc}
\sfindex  & 0 & 1 & 2 & 3 & 4 & 5 & 6 & 7 & 8 & 9 & 10& 11& 12 \\
  T    & 7 & 6 & 6 & 9 & 9 & 6 & 6 & 9 & 9 & 6 & 7 & 1 & 0 \\
  \LMS  &   & * &   &   &   & * &   &   &   & * &   &   & * \\
  \SA   & E & E & E & E & E & E & E & E & E & E & E & E & E
\end{array}
\]
After initialization:
\[
\begin{array}{cccccccccccccc}
\sfindex  & 0 & 1 & 2 & 3 & 4 & 5 & 6 & 7 & 8 & 9 & 10& 11& 12 \\
  \SA   &(\urb{U}) &(E)&(E& E & E & E& \urb{M}) &(E& E)&(E & E & E & E) \\
\end{array}
\]
Step 2. Placing all indices of LMS-characters into $\SA$:
\[
\begin{array}{cccccccccccccc}
\sfindex  & 0 & 1 & 2 & 3 & 4 & 5 & 6 & 7 & 8 & 9 & 10& 11& 12 \\
  \SA   &(\urb{12}) &(E)&(E& E & E & E& M) &(E& E)&(E & E & E & E)\\
  \SA   &(12) &(E)&(E& E & \urb{9} & \redb{1}& \redb{M}) &(E& E)&(E & E & E & E)\\
  \SA   &(12) &(E)&(E& \urb{5} & 9 & \redb{2}& \redb{M}) &(E& E)&(E & E & E & E)\\
  \SA   &(12) &(E)&(\urb{1}& 5 & 9 & \redb{3}& \redb{M}) &(E& E)&(E & E & E & E)\\
  \SA   &(12) &(E)&(E& E & \redb{1} & \redb{5}& \redb{9}) &(E& E)&(E & E & E & E)
\end{array}
\]
In the last row, we remove all $\NUnique$ symbols and counters.
\QEDB
\end{example}

\subsection{Induced sort all LMS-substrings from the sorted LMS-characters}
\label{secint:sortlms}

In this section, we sort all LMS-substrings from the sorted LMS-characters using induced sorting. Since all LMS-substrings are LMS-prefix of S-suffixes (Recall that the LMS-prefix of $\Suf(i)$ is $T[i\ldots j]$, where $j>i$ is the smallest position in $\Suf(i)$ such that $T[j]$ is an LMS character) and sorting the LMS-prefix of all suffixes from the sorted LMS-characters is the same as sorting all suffixes from the sorted LMS-suffixes (see the preliminary Section \ref{sec:pre}).

Now, we describe the details. We divide this step into two parts.
\begin{enumerate}[(1)]
  \item First, we sort the LMS-prefix of all suffixes from the sorted LMS-characters. Since this part is the same as sorting all suffixes from the sorted LMS-suffixes, we will describe the details in Section \ref{secint:sortt}.
  \item Then, we place the indices of all sorted LMS-substrings in $\SA[n-n_1\ldots n-1]$, where $n_1$ denotes the number of LMS-characters. Note that the number of LMS-characters, LMS-suffixes, and LMS-substrings are the same. Moreover, $n_1 \leq \frac{n}{2}$ since any two LMS-characters are not adjacent.
\end{enumerate}

The first part is described and proved in Section \ref{secint:sortt}.
Here, we only need to explain the second part how to place the indices of all sorted LMS-substrings in $\SA[n-n_1\ldots n-1]$. First, we need the following Observation \ref{ob:bucket} and Lemma \ref{lem:numbers}. Then we give a lemma to show that this step can be done in linear time using $O(1)$ workspace.

\begin{observation}\label{ob:bucket}
For any bucket in $\SA$, let $t$ be its bucket tail. Then $T[\SA[t]]$ is S-type if and only if $T[\SA[t]]< T[\SA[t]+1]$. Similarly, $T[\SA[h]]$ is L-type if and only if $T[\SA[h]]> T[\SA[h]+1]$, where $h$ is the bucket head.
\end{observation}

\begin{lemma}\label{lem:numbers}
If a bucket contains S-type characters, then one can scan this bucket once to compute the number of S-type characters in this bucket using $O(1)$ workspace.
\end{lemma}
\begin{proof}
We scan this bucket from its tail to its head. For the current scanning entry $\SA[i]$,
there are two possibilities:
1). If $T[\SA[i]]\geq T[\SA[i]+1]$, do nothing;
2). Otherwise, let $j$ be the smallest index such that $T[k]=T[\SA[i]]$ for any $k \in [j,\SA[i]]$.
Then we increase $num$ by $\SA[i]-j+1$. Here variable $num$ counts the number of S-type characters in this bucket and initially is $0$.
\end{proof}

\begin{lemma}
The indices of all sorted LMS-substrings can be placed in $\SA[n-n_1\ldots n-1]$ using linear time and $O(1)$ workspace.
\end{lemma}
\begin{proof}
In Step (2), we scan $\SA$ from right to left to place the indices of all LMS-substrings at the end of $\SA$.
We only need to explain how to distinguish whether $T[\SA[i]]$ is LMS-character or not when we scanning $\SA[i]$.
Note that if we can tell if $T[\SA[i]]$ is S-type or not, we can also tell if $T[\SA[i]]$ is an LMS-character or not
since $T[\SA[i]]$ is an LMS-character if and only if $T[\SA[i]]$ is S-type and $T[\SA[i]-1]>T[\SA[i]]$.
In the scanning process, when we reach a new bucket, we can see whether this bucket contains S-type characters or not
from Observation \ref{ob:bucket}.
Furthermore, if we know the number of S-type characters in this bucket (Lemma \ref{lem:numbers}),
we are done with this step since all S-type suffixes appear after the L-suffixes in any bucket (Property \ref{prop:ls}).
This step costs $O(n)$ time overall since each character is scanned at most twice.
\end{proof}

\begin{example}
Continue our example:
\[
\begin{array}{cccccccccccccc}
\sfindex  & 0 & 1 & 2 & 3 & 4 & 5 & 6 & 7 & 8 & 9 & 10& 11& 12 \\
  \LMS  &   & * &   &   &   & * &   &   &   & * &   &   & * \\
  \SA   &(12) &(E)&(E& E & 1 & 5& 9) &(E& E)&(E & E & E & E)
\end{array}
\]
(1) Sorting the LMS-prefixes of all suffixes (see Section \ref{secint:sortt}):
\[
\begin{array}{cccccccccccccc}
 \sfindex  & 0 & 1 & 2 & 3 & 4 & 5 & 6 & 7 & 8 & 9 & 10& 11& 12 \\
   \SA  &(12)&(11)&(1&5&9&2&6)&(10& 0)&(4& 8& 3 &7)
\end{array}
\]
(2) Placing the indices of all sorted LMS-substrings in $\SA[n-n_1\ldots n-1]$:
\[
\begin{array}{cccccccccccccc}
\sfindex  & 0 & 1 & 2 & 3 & 4 & 5 & 6 & 7 & 8 & 9 & 10& 11& 12 \\
  \SA   & E & E & E & E & E & E & E & E & E & \urb{12}& \urb{1} & \urb{5} & \urb{9}
\end{array}
\]
\QEDB
\end{example}

\subsection{Construct the reduced problem $T_1$}
\label{secint:gett1}

In this section, we construct the smaller recursive problem $T_1$.
We rename the sorted LMS-substrings (obtained from the previous step) using their ranks to obtain $T_1$.
Note that this step is not difficult and similar to the previous algorithms (e.g.,~\cite{Nong2009,Nong2013}).

Now, we spell out the details for this step.
Initially, all LMS-substrings are sorted in $\SA[n-n_1\ldots n-1]$.
First, we let the rank of the smallest LMS-substring corresponding to $\SA[n-n_1]$ (i.e., the LMS-substring which begins from index $\SA[n-n_1]$) be 0 (it must be the sentinel).
Then, we scan $\SA[n-n_1+1\ldots n-1]$ from left to right to compute the rank for each LMS-substring.  When scanning $\SA[i]$, we compare the LMS-substring corresponding to $\SA[i]$ and that corresponding to $\SA[i-1]$.
If they are the same,  $\SA[i]$ gets the same rank as $\SA[i-1]$.
Otherwise, the rank of $\SA[i]$ is the rank of $\SA[i-1]$ plus 1.
Since we have no extra space, we need to store the ranks in $\SA$ as well.
In particular, the rank of $\SA[i]$ is stored in $\SA[\lfloor\frac{\SA[i]}{2}\rfloor]$.
There is no conflict since any two LMS-characters are not adjacent.
Finally, we shift nonempty entries in $\SA[0\ldots n-n_1-1]$ to the head of $\SA$, so that the ranks occupy a
consecutive segment of the space.
Now, we have obtained the reduced problem $T_1$ which is stored in $\SA[0\ldots n_1-1]$.
In other words, $\SA[i]$ ($i\in [0,n_1-1]$)
stores the new name of the $i$-th LMS-substring with respect to its appearance in the input string $T$.

\begin{example}
Continue our example:
\[
\begin{array}{cccccccccccccc}
\sfindex  & 0 & 1 & 2 & 3 & 4 & 5 & 6 & 7 & 8 & 9 & 10& 11& 12 \\
  T    & 7 & 6 & 6 & 9 & 9 & 6 & 6 & 9 & 9 & 6 & 7 & 1 & 0 \\
  \LMS  &   & * &   &   &   & * &   &   &   & * &   &   & * \\
  \SA   & E & E & E & E & E & E & E & E & E & 12& 1 & 5 & 9 \\
\end{array}
\]
After scanning $\SA[n-n_1\ldots n-1]$ (which stores the sorted LMS-substrings):
\[
\begin{array}{cccccccccccccc}
\sfindex  & 0 & 1 & 2 & 3 & 4 & 5 & 6 & 7 & 8 & 9 & 10& 11& 12 \\
  \SA   & \urb{1} & E & \urb{1} & E & \urb{2} & E & \urb{0} & E & E & 12& 1 & 5 & 9 \\
\end{array}
\]
Finally, we get $T_1$ stored in $\SA[0\ldots n_1-1]$ by shifting nonempty items in $\SA[0\ldots n-n_1-1]$ to the head of $\SA$.
\[
\begin{array}{cccccccccccccc}
\sfindex  & 0 & 1 & 2 & 3 & 4 & 5 & 6 & 7 & 8 & 9 & 10& 11& 12 \\
  \SA   & \urb{1} & \urb{1} & \urb{2} & \urb{0} & E & E & E & E & E & 12& 1 & 5 & 9 \\
\end{array}
\]
Note that $T_1 = ``1120"$ which corresponds to the LMS-substrings $\{``66996", ``66996", ``6710", ``0"\}$.
\QEDB
\end{example}

First, we give an observation which helps us to identify the S-type and L-type characters of $T$. Then we obtain the following lemma
which implies that $T_1$ can be obtained in linear time.

\begin{observation}\label{ob:nexts}
For any index $i$ of $T$, let $j\in[i+1,n-1]$ be the smallest index such that $T[j]<T[j+1]$
(So $T[j]$ is S-type).
Furthermore let $k\in[i+1,j]$ be the smallest index such that $T[l]=T[j]$ for any $k\leq l\leq j$.
Then $T[k]$ is the first S-type character after index $i$.  Moreover, all characters between $T[i]$ and $T[k]$ are L-type, and characters between $T[k]$ and $T[j]$ are S-type.
\end{observation}

\begin{lemma}\label{lem:ct1}
$T_1$ can be obtained using $O(n)$ time and $O(1)$ workspace.
\end{lemma}
\begin{proof}
For the workspace part, it is obvious since we do not use any extra space beyond $\SA$ in the above step.
For the time part, we only need to explain the running time of the comparison process.
When we compare $\SA[i]$ and $\SA[i-1]$, we can know the length of these two LMS-substrings (indicated by $\SA[i]$ and $\SA[i-1]$) from the Observation \ref{ob:nexts}. Note that each character of $T$ is scanned at most twice
since it is only scanned when identifying its length and its adjacent predecessor LMS-substring.
Thus the comparison process takes $O(n)$ time because the total length of all LMS-substrings is less than $2n$.
\end{proof}

\subsection{Sort the LMS-suffixes by solving $T_1$ recursively}
\label{secint:sorttlms}

In this section, we sort all LMS-suffixes and place their indices in the tail of their corresponding buckets in $\SA$.
This can be done as follows:

\begin{enumerate}
\item We first solve $T_1$ recursively.
From Section \ref{secint:gett1}, $T_1$ is stored in $\SA[0\ldots n_1-1]$.
We define $\SA_1$ to be $\SA[n-n_1\ldots n-1]$ and use $\SA_1$ to store the output of the subproblem $T_1$.

\item Now, we put all indices of LMS-suffixes in $\SA$. First we move $\SA_1$ to $\SA[0\ldots n_1-1]$ (i.e., move $\SA[n-n_1\ldots n-1]$ to $\SA[0\ldots n_1-1]$). Then we scan $T$ from right to left. For every LMS-character $T[i]$, place $i$ (i.e., index of $\Suf(i)$) in the tail of $\SA$.

\item For notational convenience, we define
$\LMS[0\ldots n_1]\triangleq\SA[n-n_1\ldots n-1]$.
Now, we obtain the sorted order of all LMS-suffixes
of the original string $T$ by
letting $\SA[i] = \LMS[\SA[i]]$ for all $i\in[0,n_1-1]$.

\item Finally, we scan $\SA[0\ldots n_1-1]$ once more from right to left, and move the indices of LMS-suffixes in the same bucket to the tail of its bucket and clear other entries. This is easy to do since each S-type $T[i]$ (after the renaming step in Section \ref{secint:renamet}) has pointed to the tail of its bucket.
\end{enumerate}

\begin{example}
Continue our example:
\[
\begin{array}{cccccccccccccc}
\sfindex  & 0 & 1 & 2 & 3 & 4 & 5 & 6 & 7 & 8 & 9 & 10& 11& 12 \\
  \SA   & 1 & 1 & 2 & 0 & E & E & E & E & E & E & E & E & E
\end{array}
\]

Step 1.
Solve $T_1$ recursively:
\[
\begin{array}{cccccccccccccc}
\sfindex  & 0 & 1 & 2 & 3 & 4 & 5 & 6 & 7 & 8 & 9 & 10& 11& 12 \\
  \SA   & 1 & 1 & 2 & 0 & E & E & E & E & E &\urb{3}&\urb{0}&\urb{1}&\urb{2}
\end{array}
\]

Step 2.
After move $\SA_1$ to $\SA[0\ldots n_1-1]$ and put all LMS-suffixes in $\SA$:
\[
\begin{array}{cccccccccccccc}
\sfindex  & 0 & 1 & 2 & 3 & 4 & 5 & 6 & 7 & 8 & 9 & 10& 11& 12 \\
  \SA   &\urb{3}&\urb{0}&\urb{1}&\urb{2}& E& E& E&E& E&\urb{1} & \urb{5} & \urb{9} & \urb{12}
\end{array}
\]

Step 3.
Get all sorted LMS-suffixes:
\[
\begin{array}{cccccccccccccc}
\sfindex  & 0 & 1 & 2 & 3 & 4 & 5 & 6 & 7 & 8 & 9 & 10& 11& 12 \\
  \SA   &\urb{12}&\urb{1}&\urb{5}& \urb{9}& E& E& E&E& E&1 & 5 & 9 & 12
\end{array}
\]

Step 4.
Move the indices of LMS-suffixes in the
same bucket to the tail of its bucket and clear other entries:
\[
\begin{array}{cccccccccccccc}
\sfindex  & 0 & 1 & 2 & 3 & 4 & 5 & 6 & 7 & 8 & 9 & 10& 11& 12 \\
  \SA   &(\urb{12})&(E)&(E  & E & \urb{1}& \urb{5}& \urb{9})&(E& E)&(E & E & E & E)
\end{array}
\]
\QEDB
\end{example}

\begin{lemma}\label{lem:sortt1}
All LMS-suffixes can be sorted by solving the reduced problem $T_1$ recursively and placed in the tail of their corresponding buckets in $\SA$ using $O(n)$ time and $O(1)$ workspace.
\end{lemma}
\begin{proof}
The time and space used in this step are easy to verify.
We only show the correctness of this step.
Each character of $T_1$ corresponds to an LMS-substring of $T$ and this character is the rank of the corresponding sorted LMS-substring. Hence, the lexicographical order of LMS-suffixes of $T$ is the same as the order of suffixes in $T_1$.
\end{proof}

\subsection{Induced sort all suffixes of $T$}
\label{secint:sortt}

Now, we sort all suffixes of $T$ from the sorted
LMS-suffixes using induced sorting with our interior counter trick (Note that this step is the same as what we did
in Section \ref{secint:sortlms}. Now, we describe the details here). First, we induced-sort the order of
all L-suffixes from LMS-suffixes. Then we induce the order of S-suffixes from the L-suffixes.
Now, we show how to carry out these steps with the desired optimal time and space.

\topic{Step 1. Induced sort all L-suffixes from the sorted LMS-suffixes}
Some details of this step is similar to our previous step in Section \ref{secint:putlms} where we introduced our interior counter trick. We divide this step into two parts as follows:
\begin{enumerate}[(1)]
      \item First initialize $\SA$:
      We scan $T$ from right to left. For every $T[i]$ which is L-type, do the following:
             \begin{enumerate}[(i)]
                  \item If $\SA[T[i]]= \Empty$, let $\SA[T[i]]= \Unique$ (unique L-type character in this bucket).
                  \item If $\SA[T[i]] = \Unique$, let $\SA[T[i]] = \NUnique$ (the number of L-type characters in this bucket is at least $2$).
                  \item Otherwise do nothing.
             \end{enumerate}
      \item Then we scan $\SA$ from left to right to sort all the L-suffixes.
          \begin{enumerate}[(i)]
            \item If $\SA[i]=\Empty$, do nothing.
            \item If $\SA[i]$ is an index, we let $j = \SA[i]-1$. Then, if \suf{j} is L-suffix (this can be identified in constant time from the following Lemma \ref{lem:knowtype}),
          we place \suf{j} into the LF-entry (recall that LF-entry denotes the leftmost free entry in its bucket) of its bucket and increase the counter by one.
         \item If $\SA[i]=\NUnique$, which means $\SA[i]$ is the head of its bucket, and this bucket has at least two L-suffixes which are not sorted, we use $\SA[i]$ and $\SA[i+1]$ as the bucket head (the symbol $\NUnique$) and
         the counter of this bucket, respectively. Then we skip these two entries and continue to scan $\SA[i+2]$.
          \end{enumerate}
    \end{enumerate}
Now, all L-suffixes have been sorted. Note that we still need to scan $\SA$ once more to free these positions occupied by $\NUnique$ and counters. After this, the indices of all L-suffixes are in their final positions in $\SA$.

\topic{Step 2. Remove LMS-Suffixes from $\SA$}
We can use a trick similar to the previous Step 2 in Section \ref{secint:putlms},
i.e., placing the indices of LMS-characters into $\SA$.
The difference is that instead of placing the actual LMS-characters, we place the $\Empty$ symbol instead.
Also note that we do not delete the sentinel since it must be in the final position.
Now, $\SA$ contains only all L-suffixes and the sentinel, and all of them are in their final positions in $\SA$.

\topic{Step 3. Induced sort all S-suffixes from the sorted L-suffixes}
Now, this step is completely symmetrical to the above Step 1 (Sort all L-suffixes using induced sorting). We use S-type and RF-entry instead of L-type and LF-entry, and we do not repeat the details here.

\begin{example}
Continue our example:
\[
\begin{array}{cccccccccccccc}
\sfindex  & 0 & 1 & 2 & 3 & 4 & 5 & 6 & 7 & 8 & 9 & 10& 11& 12 \\
T    & 7 & 6 & 6 & 9 & 9 & 6 & 6 & 9 & 9 & 6 & 7 & 1 & 0 \\
\type   & L & S & S & L & L & S & S & L & L & S & L & L & S \\
\SA   &(12)&(E)&(E & E & 1 & 5 &9)&(E& E)&(E & E & E & E)
\end{array}
\]
Step 1. Induced-sort all L-suffixes from the sorted LMS-suffixes:\\
(1) After initialization:
\[
\begin{array}{cccccccccccccc}
\sfindex  & 0 & 1 & 2 & 3 & 4 & 5 & 6 & 7 & 8 & 9 & 10& 11& 12 \\
\SA   &(12)&(\urb{U})&(E & E & 1 & 5 &9)&(\urb{M}& E)&(\urb{M} & E & E & E)
\end{array}
 \]
(2) Scan $\SA$ from left to right to sort all L-suffixes:
\[
\begin{array}{cccccccccccccc}
\sfindex  & 0 & 1 & 2 & 3 & 4 & 5 & 6 & 7 & 8 & 9 & 10& 11& 12 \\
 \SA   &(\vrb{12})&(\urb{11})&(E & E & 1 & 5 &9)&(M& E)&(M & E & E & E)\\
  \SA   &(12)&(\vrb{11})&(E & E & 1 & 5 &9)&(\urb{10}& E)&(M & E & E & E) \\
  \SA   &(12)&(11)&(E & E & \vrb{1} & 5 &9)&(10& \urb{0})&(M & E & E & E)\\
  \SA   &(12)&(11)&(E & E & 1& \vrb{5} &9)&(10& 0)&(\urb{M}&\urb{1}&\urb{4} & E)\\
  \SA   &(12)&(11)&(E & E & 1& 5 &\vrb{9})&(10& 0)&(\urb{M} & \urb{2} & 4 & \urb{8})\\
  \SA   &(12)&(11)&(E & E & 1& 5 &9)&(10& 0)&(M & 2 & \vrb{4} & 8)\\
  \SA   &(12)&(11)&(E & E & 1& 5 &9)&(10& 0)&(\vrb{4}& \urb{8} & \urb{3}& \urb{E})\\
  \SA   &(12)&(11)&(E & E & 1& 5 &9)&(10& 0)&(4&\vrb{8}&3&\urb{7})\\
\end{array}
\]
\qquad Note that the third to last line is the case (iii). So we skip these two entries (i.e., `$M$' and `$2$').\\
Step 2. Remove LMS-Suffixes from $\SA$:
\[
\begin{array}{cccccccccccccc}
\sfindex  & 0 & 1 & 2 & 3 & 4 & 5 & 6 & 7 & 8 & 9 & 10& 11& 12 \\
\SA &(12)&(11)&(E&E&\urb{E}&\urb{E}&\urb{E})&(10& 0)&(4& 8& 3 &7)
\end{array}
\]
Step 3. Induced-sort all S-suffixes from the sorted L-suffixes:\\
(1) After initialization:
\[
\begin{array}{cccccccccccccc}
\sfindex  & 0 & 1 & 2 & 3 & 4 & 5 & 6 & 7 & 8 & 9 & 10& 11& 12 \\
\SA &(12)&(11)&(E&E&E&E&\urb{M})&(10& 0)&(4& 8& 3 &7)
\end{array}
\]
(2) Scan $\SA$ from right to left to sort all S-suffixes:
\[
\begin{array}{cccccccccccccc}
\sfindex  & 0 & 1 & 2 & 3 & 4 & 5 & 6 & 7 & 8 & 9 & 10& 11& 12 \\
   \SA  &(12)&(11)&(E& E& \urb{6}&\urb{1}&\urb{M})&(10& 0)&(4& 8& 3 &\vlrb{7}) \\
   \SA  &(12)&(11)&(E& \urb{2}& 6&\urb{2}&\urb{M})&(10& 0)&(4& 8& \vlrb{3} &7) \\
   \SA  &(12)&(11)&(\urb{9}& 2& 6&\urb{3}&\urb{M})&(\vlrb{10}& 0)&(4& 8& 3 &7)  \\
   \SA  &(12)&(11)&(9& 2&\vlrb{6}&3&M)&(10& 0)&(4& 8& 3 &7)  \\
   \SA  &(12)&(11)&(\urb{E}&\urb{5}&\urb{9}&\urb{2}&\vlrb{6})&(10& 0)&(4& 8& 3 &7) \\
   \SA  &(12)&(11)&(\urb{1}&5&9&\vlrb{2}&6)&(10& 0)&(4& 8& 3 &7) \\
   \SA  &(12)&(11)&(1&5&9&2&6)&(10& 0)&(4& 8& 3 &7)
\end{array}
\]
\QEDB
\end{example}

In order to show the time used in this step, we need the following useful lemma to identify the L/S-type information in the induced-sorting step.

\begin{lemma}\label{lem:knowtype}
When we scan $\SA$ from left to right to induced-sort L-suffixes, for each scanned $\SA[i]$, we want to identify the type of $\Suf(\SA[i]-1)$. If $T[\SA[i]-1]\neq T[\SA[i]]$, the type of $\Suf(\SA[i]-1)$ can be obtained immediately.
Otherwise $T[\SA[i]-1]= T[\SA[i]]$ (this case $\Suf(\SA[i]-1)$ belongs to the current scanned bucket $T[\SA[i]]$).
If all L-suffixes of $T$ belonging to bucket $T[\SA[i]]$ are not already sorted, then $\Suf(\SA[i]-1)$ is an L-suffix.
\end{lemma}
\begin{proof}
According to Property \ref{prop:ls}, in any bucket, the S-suffixes always appear after the L-suffixes in $\SA$. Besides, it is obvious that every suffix of $T$ is considered exactly once. Furthermore,
we can distinguish whether all L-suffixes of $T$ belonging to the
current bucket $T[\SA[i]]$ are already sorted or not by scanning the current bucket once, when we are reaching a new bucket. Combining these observations, the lemma is proved.
\end{proof}

\begin{lemma}
Given all sorted LMS-suffixes of $T$, all suffixes can be sorted correctly using $O(n)$ time and $O(1)$ workspace according to the induced sorting steps.
\end{lemma}
\begin{proof}
For the correctness: we can sort all L-suffixes correctly from the sorted LMS-suffixes using induced sorting step from Lemma \ref{lem:lmsl} and we can sort all S-suffixes correctly from the sorted L-suffixes using induced sorting step from Lemma \ref{lem:ls}.
\end{proof}

Now, we obtain the following theorem for our optimal in-place algorithm.
\begin{theorem}\label{thm:int}
Our Algorithm takes $O(n)$ time and $O(1)$ workspace to compute the suffix array of string $T$ over integer alphabets with $|\Sigma| \leq n$.
\end{theorem}
\begin{proof}
The time complexity simply follows from the recursion $T(n) = T(n/2)+ O(n)=O(n)$.
For the workspace, every step in our algorithm uses $O(1)$ workspace
and different steps can reuse the $O(1)$ workspace too.
In the recursive subproblem, we can also reuse the $O(1)$ workspace.
\footnote{If one worries the $O(\log n)$ workspace in
the recursion, one can use the highest bits in $\SA$ (i.e., $n$ bits) to store them since the size of the reduced sub-problem is no larger than $n/2$.}
\end{proof}

\section{Suffix Sorting for Read-only Integer Alphabets}
\label{sec:roint}

\subsection{Framework}
\label{sec:intfw}
First, we define some notations.
Let $n_L$ and $n_S$ denote the number of L-suffixes and S-suffixes, respectively.
Let $n_1$ denote the length of the reduced problem $T_1$, i.e.,
$n_1$ equals to the number of LMS-suffixes (Case 1) or LML-suffixes (Case 2).
Now, we describe the framework of our algorithm as follows:
\begin{enumerate}
  \item If $n_L \leq n_S$ (i.e., the number of L-suffixes is no larger than that of S-suffixes),
  then:
    \begin{enumerate}[(1)]
      \item (Section \ref{sec:roint1}) Sort all LMS-characters of $T$.\\
          We use counting sort to sort all LMS-characters of $T$ in $\SA[n-n_1\ldots n-1]$.
          In the counting sort step, we use $\SA[0\ldots n/2]$ as the temporary space (counting array). After this step, all indices of the sorted LMS-characters are stored in $\SA[n-n_1\ldots n-1]$. Note that this step is different from our previous algorithm in Section \ref{sec:int} since we sort and store the LMS-characters in the end of $\SA$ instead of their corresponding buckets, because $T$ is read-only now.

      \item (Section \ref{sec:roint2}) Induced sort all LMS-substrings from the sorted LMS-characters.\\
            This induced-sorting step is the same as Step (4) below where we induced-sort all suffixes from the sorted LMS-suffixes.
            Thus, we only describe the details of this step in Section \ref{sec:roint2}.
            After this step, all indices of the sorted LMS-substrings are stored in $\SA[n-n_1\ldots n-1]$.

      \item (Section \ref{sec:roint3}) Construct and solve the reduced problem $T_1$ from the sorted LMS-substrings.\\
          We construct the reduced problem $T_1$ using the ranks of all sorted LMS-substrings which are stored in $\SA[n-n_1\ldots n-1]$,
          where the ranks of LMS-substrings correspond to the lexicographical order of the sorted LMS-substrings (this construction step is the same as that in Section \ref{secint:gett1}).
          Then we get the reduced problem $T_1$ in $\SA[0\ldots n_1-1]$ and solve $T_1$ recursively to obtain the sorted LMS-suffixes. In the recursive step, we use $\SA_1=\SA[n-n_1\ldots n-1]$ as the output space for $T_1$.
          After this step, all indices of the sorted LMS-suffixes are stored in $\SA[n-n_1\ldots n-1]$.
      \item (Section \ref{sec:roint2}) Induced sort all suffixes of $T$ from the sorted LMS-suffixes ($T_1$).\\
          We induced-sort all suffixes of $T$ from the sorted LMS-suffixes which are stored in $\SA[n-n_1\ldots n-1]$.
          In the sorting step, we use the \emph{interior counter trick}, as in Section~\ref{sec:int},
          which can implicitly represent the dynamic LF/RF-entry information in $\SA$. Besides, we provide a \emph{pointer data structure} which can represent the bucket heads/tails in $\SA$.
          Combining these two techniques, we can remove the workspace needed by the bucket array.
          Due to the space required by our pointer data structure is nonconstant, we divide this induced sorting step into two stages to address this issue.
          Note that this step is the main technical part of our optimal in-place algorithm.
          After this step, all indices of the suffixes of $T$ are sorted and stored in $\SA[0\ldots n-1]$.
    \end{enumerate}
  \item Otherwise, execute the above steps switching the role of LMS with LML.
\end{enumerate}

Without loss of generality, we assume that $n_L\leq n_S$.
Note that we compare the number of L-suffixes and S-suffix at the beginning since we need half of the space of $\SA$ to construct our pointer data structure for induced-sorting the L-suffixes (from the sorted LMS-suffixes) and S-suffixes (from the sorted L-suffixes) in Step (4). The in-place implementation of this induced sorting step is the main technical part of our optimal in-place algorithm.
Note that the empty space is enough since the number of LMS-suffixes (i.e., $n_1$) and L-suffixes (i.e., $n_L$) both are less than or equal to $n/2$, where $n_1\leq n/2$ since any two LMS-characters are not adjacent by the definition of LMS-characters, and $n_L\leq n/2$ since $n_L\leq n_S$.
Note that for previous algorithms (e.g.,~\cite{Nong2009,Nong2013}),
they do not need the comparison at the beginning
since they use the bucket array (which needs $|\Sigma|$ words workspace) in the induced sorting step.
Here, we construct the pointer data structure and combine our interior counter trick to remove the workspace.

Now, we describe the details of our in-place algorithm in the following sections.

\eat{
\begin{algorithm}[!htb]
\caption{Optimal in-place suffix sorting algorithm for \emph{read-only} integer alphabets}
\label{alg:2}
\begin{algorithmic}[1]
\REQUIRE $T$;
\IF{$n_L\leq n_S$}
\STATE Use counting sort to sort all LMS-characters of $T$ and store them in $\SA[n-n_1\ldots n-1]$,
using $\SA[0\ldots n/2]$ as the temporary space (counting array);
\STATE Induced sort all LMS-substrings from the sorted LMS-characters and store them in $\SA[n-n_1\ldots n-1]$, using our interior counter trick and pointer data structure;
\STATE Construct the reduced recursive problem $T_1$ using the ranks of the sorted LMS-substrings and store $T_1$ in $\SA[0\ldots n_1-1]$, using naive comparison to obtain the ranks;
\STATE Let $\SA_1$ denote $\SA[n-n_1\ldots n-1]$ which is reused as the output space for problem $T_1$;
\IF{all characters in $T_1$ are unique}
\STATE Directly construct $\SA_1$ for $T_1$;
\ELSE \STATE Obtain $\SA_1$ by solving $T_1$ recursively, reusing the space of $\SA$ by induction;
\ENDIF
\STATE Sort all LMS-suffixes of $T$ from $\SA_1$ and store them in $\SA[n-n_1\ldots n-1]$, simply replacing $\SA_1$ with the indices of LMS-suffixes;
\STATE Induced sort all suffixes of $T$ from the sorted LMS-suffixes, using our interior counter trick and pointer data structure (this step is the same as Line 3);
\ELSE \STATE Execute above steps (i.e., Line 2--12) switching the role of LMS with LML;
\ENDIF
\RETURN $\SA$
\end{algorithmic}
\end{algorithm}
}

\subsection{Sort all LMS-characters of $T$}
\label{sec:roint1}

In this section, we sort all LMS-characters of $T$ and place their indices in $\SA[n-n_1\ldots n-1]$.
Recall that $n_1$ denotes the number of LMS-characters.

Now, we describe the details.
Since $|\Sigma|= O(n)$, we can assume that $|\Sigma|\leq dn$ for some constant $d$.
We divide the LMS-characters of $T$ into $2d$ partitions and sort each partition one by one.
The partition $i$ contains the LMS-characters which belong to
$\left[\frac{i|\Sigma|}{2d} +1,\frac{(i+1)|\Sigma|}{2d}\right]$, for $0\leq i<2d$.
We use $m_i$ to denote the number of LMS-characters in partition $i$.
Then for each partition $i$, we use the standard \emph{counting sort}
(see e.g.,~\cite[Chap. 8]{thomas2001introduction})
to sort these $m_i$ LMS-characters (the LMS-characters can be identified by scanning $T$ once from right to left).
Concretely, we use $\SA[0\ldots n/2]$ as the temporary counting array, and use $\SA[n/2+\sum_{j=0}^{i-1}m_j+1\ldots n/2+\sum_{j=0}^i m_j]$ as the output array.
After this counting sort step, the indices of these $m_i$ sorted LMS-characters have been placed in
$\SA[n/2+\sum_{j=0}^{i-1}m_j+1\ldots n/2+\sum_{j=0}^i m_j]$.

Note that we can use the counting sort step for each partition.
Because the gap of each partition is $\frac{|\Sigma|}{2d}\le \frac{dn}{2d}=\frac{n}{2}$,
the space of $\SA[0\ldots n/2]$ is enough for the temporary counting array (its size equals to the gap) of counting sort step.
It is not hard to see that the sorting step takes $O(n)$ time and uses $O(1)$ workspace
since we only make $2d$ times of counting sort steps (each step takes linear time).

After sorting all $2d$ partitions, all indices of the sorted LMS-characters are placed in $\SA[n/2+1,n/2+\sum_{j=0}^{2d-1}m_j]$ (i.e., $\SA[n/2+1,n/2+n_1]$).
Then we move them to $\SA[n-n_1\ldots n-1]$, which can be easily done in
linear time and $O(1)$ workspace.

\subsection{Construct and solve the reduced problem $T_1$ from the sorted LMS-substrings}
\label{sec:roint3}

\topic{Construct the reduced problem $T_1$}
We construct the reduced problem $T_1$ using the ranks of all sorted LMS-substrings which are stored in $\SA[n-n_1\ldots n-1]$ from the Step (2) (see the framework in Section \ref{sec:intfw}),
where the ranks of LMS-substrings are corresponding to the lexicographical order of the sorted LMS-substrings.

Note that this construction step is exactly the same as that in Section \ref{secint:gett1}.
Thus we omit the details and recall the same lemma as follows.

\begin{lemma}
$T_1$ can be constructed using $O(n)$ time and $O(1)$ workspace.
\end{lemma}

\topic{Solve $T_1$ recursively}
\label{secint:sortt1}
Now, we sort all LMS-suffixes by solving $T_1$ recursively and place their indices in the tail of $\SA$ (i.e. $\SA[n-n_1\ldots n-1]$).
This step is carried out as follows (similar to the Section~\ref{secint:sorttlms} except  Step 4 below):

\begin{enumerate}
\item We first solve $T_1$ recursively.
Recall that $T_1$ is stored in $\SA[0\ldots n_1-1]$.
We define $\SA_1$ to be $\SA[n-n_1\ldots n-1]$ and use $\SA_1$ to store the output of the subproblem $T_1$.

\item Now, we put all indices of LMS-suffixes in $\SA$. First we move $\SA_1$ to $\SA[0\ldots n_1-1]$ (i.e., move $\SA[n-n_1\ldots n-1]$ to $\SA[0\ldots n_1-1]$). Then we scan $T$ from right to left. For every LMS-character $T[i]$, place $i$ (i.e., index of $\Suf(i)$) in the tail of $\SA$.

\item For notational convenience, we define
$\LMS[0\ldots n_1]\triangleq\SA[n-n_1\ldots n-1]$.
Now, we obtain the sorted order of all LMS-suffixes
of the original string $T$ by
letting $\SA[i] = \LMS[\SA[i]]$ for all $i\in[0,n_1-1]$.

\item Finally, we finish this step by moving $\SA[0\ldots n_1-1]$ to $\SA[n-n_1\ldots n-1]$.
Now, all indices of the sorted LMS-suffixes are stored in $\SA[n-n_1\ldots n-1]$.
\end{enumerate}

\begin{lemma}\label{lem:rosortt1}
All LMS-suffixes can be sorted by solving the reduced problem $T_1$ recursively and placed in the tail of $\SA$ using $O(n)$ time and $O(1)$ workspace.
\end{lemma}
\begin{proof}
The time and space used in this step are easy to verify.
We only show the correctness of this step.
Each character of $T_1$ corresponds to an LMS-substring of $T$ and this character is the rank of the corresponding sorted LMS-substring. Hence, the lexicographical order of LMS-suffixes of $T$ is the same as the order of suffixes in $T_1$.
\end{proof}

\vspace{2mm}
\subsection{Induced sort all suffixes of $T$ from the sorted LMS-suffixes}
\label{sec:roint2}

In this section, we show how to induced-sort all suffixes from the sorted LMS-suffixes.
All indices of the sorted LMS-suffixes have been placed in $\SA[n-n_1\ldots n-1]$ from the previous step (see Lemma~\ref{lem:rosortt1}).
Note that this step is the main technical part of our in-place algorithm.
We develop two techniques
(\emph{interior counter trick} and \emph{pointer data structure}) for implementing this induced sorting in-place.

Let $\SA_L=\SA[0\ldots n_L-1]$ and $\SA_S=\SA[n_L\ldots n-1]$.
Recall that $n_S$ and $n_L$ denote the number of S-suffixes and L-suffixes, respectively.
Also note that $n_L +n_S =n$.
First, we sort all $n_L$ L-suffixes from the sorted LMS-suffixes which are stored in $\SA[n-n_1\ldots n-1]$ and store the sorted L-suffixes in $\SA_L$.
Then, we sort all $n_S$ S-suffixes from the sorted L-suffixes and store the sorted S-suffixes in $\SA_S$.
Finally, we merge the sorted L-suffixes (stored in $\SA_L$) and S-suffixes (stored in $\SA_S$) to sort all suffixes in $\SA[0\ldots n-1]$.

Note that sorting the $n_L$ L-suffixes from the sorted LMS-suffixes is totally symmetrical as
sorting the $n_S$ S-suffixes from the sorted L-suffixes, as stated in Section~\ref{sec:pre}.
Thus, we only need to show the details of how to sort all $n_S$ S-suffixes from the sorted L-suffixes which has already been stored in $\SA_L$, and store the sorted S-suffixes in $\SA_S$.
Before we show the details, we briefly recall the original induced sorting step here (we have introduced it in Section~\ref{sec:pre} and provided a running example in Appendix~\ref{app:inducedsort}.).
Now, we recall the symmetrical case here,
i.e., inducing the order of S-suffixes from the sorted L-suffixes.

\vspace{1mm}
\topic{Inducing the order of S-suffixes from the sorted L-suffixes}
We scan $\SA$ from right to left (i.e., from $\SA[n-1]$ to $\SA[0]$).
When we scan $\SA[i]$, let $j=\SA[i]-1$.
If $T[j]$ is S-type, i.e., $\Suf(j)$ is an S-suffix (indicated by the type array),
we place the index of $\Suf(j)$ (i.e. $j$) into the RF-entry of bucket $T[j]$,
and then let the RF-pointer of this bucket $T[j]$ point to the next free entry.
If $T[j]$ is L-type, we do nothing (since all L-suffixes are sorted in the correct positions).
The RF-pointers are maintained by the bucket array.

In order to obtain the in-place algorithm, we need to show how to remove the workspace needed by the bucket array and type array in the induced sorting step.
Briefly speaking,
the purpose of the pointer data structure is to indicate the bucket tails of S-suffixes,
and the purpose of the interior counter trick is to maintain the RF-pointers of the buckets dynamically.
Thus, for a query of RF-entry for $\Suf(j)$ in bucket $T[j]$, we know the tail of the bucket $T[j]$ from the pointer data structure in constant time (Lemma~\ref{lem:intpds}), then we use the interior counter trick to indicate the RF-entry in this bucket (Lemma~\ref{lem:inttrick}).
For removing the type array, we use the Lemma~\ref{lem:idls} to identify the L/S-suffixes in the induced sorting step.

Now, we describe the details.
First, we introduce our interior counter trick,
extending the one from Section~\ref{sec:int}.
Here, we assume that the tail of the bucket of any S-suffix is known (which is indicated by the Lemma~\ref{lem:intpds}).

\newpage
\vspace{2mm}
\topic{Interior counter trick}
Note that the buckets of the S-suffixes we discussed in this section are in $\SA_S=\SA[n_L\ldots n-1]$, since we already have placed the sorted L-suffixes in $\SA_L=\SA[0\ldots n_L-1]$.
Thus, we only need to sort all S-suffixes to their corresponding buckets in $\SA_S$ and the buckets only contains S-suffixes now.

Here we only describe the details of interior counter trick for one bucket since other buckets are the same.
Note that we assume that the tail of the bucket of any S-suffix is known (Lemma~\ref{lem:intpds}).
To simplify the representation, we assume the bucket from index $0$ to index $m-1$ of $\SA_S$,
where $m$ is the size of this bucket (i.e. the number of S-suffixes in this bucket is $m$).
We only describe the case where $m>3$ since other cases with $m\le 3$ are similar and simpler.
We define five special symbols $\BH$ (head of the bucket), $\BT$ (tail of the bucket), $\Ept$ (Empty), $\Rone$ (one remaining S-suffix) and $\Rtwo$ (two remaining S-suffixes)
\footnote{
	Recall that the special symbol is only used to simplify the argument. See Appendix \ref{app:special} for the details.
	}.

First, we use three special symbols to initialize this bucket,
i.e., let $\SA_S[0]=\BH$, $\SA_S[m-2]=\Ept$ and $\SA_S[m-1]=\BT$.
Let $S_i$ denote the index of the $i$-th S-suffix which needs to be placed into the RF-entry of this bucket.
Now, we describe how to place the indices of these $m$ S-suffixes into the RF-entry of this bucket one by one. 
We distinguish the following four cases:

\begin{enumerate}[(1)]
    \item If $\SA_S[m-1]=\BT$, and $\SA_S[m-2]=\Ept$ or $\SA_S[m-\SA_S[m-2]-3]\neq\BH$: In this case, we place the index of the current S-suffix (i.e., $S_i$) into the RF-entry of this bucket, where $1\leq i\leq m-3$.
        Concretely, we know the position of the tail of this bucket in $\SA_S$, i.e., $m-1$ according to the assumption.
        Then, we use $\SA_S[m-2]$ as the counter to denote the number of the indices of S-suffixes has been placed so far.
        Note that the RF-entry of this bucket is pointed by this counter (i.e. RF-pointer).
        Thus, we can place the index of the current S-suffix ($S_i$) into the RF-entry of this bucket in constant time, and then update the counter $\SA_S[m-2]$.
    \item If $\SA_S[m-1]=\BT$ and $\SA_S[m-\SA_S[m-2]-3]=\BH$: In this case, we place the index of the third to last S-suffix (i.e. $S_{m-2}$) into the RF-entry of this bucket.
        Concretely, we shift the previous $m-3$ S-suffixes which stored in $\SA_S[1,\ldots,m-3]$ to $\SA_S[2,\ldots,m-2]$.
        Then, we place $S_{m-2}$ into $\SA_S[1]$ and let $\SA_S[m-1]=\Rtwo$.
        This step takes $O(m)$ time since we shift $m-3$ S-suffixes.
    \item If $\SA_S[m-1]=\Rtwo$: In this case, we place the index of the second to last S-suffix (i.e. $S_{m-1}$) into the RF-entry of this bucket.
        We shift the previous $m-2$ S-suffixes which stored in $\SA_S[1,\ldots,m-2]$ to $\SA_S[2,\ldots,m-1]$.
        Then, we place $S_{m-1}$ into $\SA_S[1]$ and let $\SA_S[0]=\Rone$.
        This step takes $O(m)$ time since we shift $m-2$ S-suffixes.
    \item Otherwise: In this case, we place the index of the last S-suffix (i.e. $S_{m}$) into the RF-entry of this bucket. First, we know the tail of the bucket indicated by our pointer data structure in constant time.
        Then, we search the entries before the tail one by one until that we find the special symbol $\Rone$. We let this entry to be $S_{m}$.
        This step takes $O(m)$ time since we search $m-1$ S-suffixes.
\end{enumerate}
In order to demonstrate these four cases more clearly, we also provide a demonstration as follows:

\[
\begin{array}{cccccccccc}
\mathsf{Index}    & 0 & 1 & \ldots  & m-3 & m-2 & m-1 \\
\SA_S     & \SA_S[0] & \SA_S[1] & \ldots  & \SA_S[m-3] & \SA_S[m-2] & \SA_S[m-1] \\
\mathsf{After~~initialization:} \\
\SA_S     & \urb{\BH} & \SA_S[1] & \ldots  & \SA_S[m-3] & \urb{\Ept} & \urb{\BT} \\
\mathsf{Case}~(1):\\
\SA_S     & \BH & \SA_S[1] & \ldots  & \urb{S_1} & \urb{1} & \BT \\
\vdots                            &&&&&\vdots& \\
\SA_S     & \BH & \urb{S_{m-3}} & \ldots  & S_1 & \urb{m-3} & \BT \\
\mathsf{Case}~(2):\\
\SA_S     & \BH & \urb{S_{m-2}} & \ldots  & \urb{S_2} & \urb{S_1} & \urb{\Rtwo} \\
\mathsf{Case}~(3):\\
\SA_S     & \urb{\Rone} & \urb{S_{m-1}} & \ldots  & \urb{S_3} & \urb{S_2} & \urb{S_1} \\
\mathsf{Case}~(4):\\
\SA_S     & \urb{S_m} & S_{m-1} & \ldots  & S_3 & S_2 & S_1 \\
\end{array}
\]
\vspace{1mm}

Note that this step uses $O(1)$ workspace since there is no bucket array and type array, and the space needed by our interior counter trick and pointer data structure is in $\SA_S$.
The purpose of the interior counter trick is to dynamic maintain the RF-pointers of the buckets.
E.g., for a query of RF-entry for $\Suf(j)$ in bucket $T[j]$, first we know the tail of the bucket $T[j]$ by the assumption, then we use the interior counter trick to indicate the RF-entry in this bucket.
We have the following lemma.
\begin{lemma}\label{lem:inttrick}
If the tail of the bucket of any S-suffix is known, one can sort the S-suffixes from the sorted L-suffixes using the induced sorting step with the interior counter trick in linear time and $O(1)$ workspace.
\end{lemma}

Note that in the induced sorting step, one uses the type array to identify whether the $\Suf(j)$ is S-suffix or not.
For removing the type array, we use the following Lemma~\ref{lem:idls} to identify the type of the L- or S-suffix in the induced-sorting step.
\begin{lemma}\label{lem:idls}
If $T[j]\neq T[\SA[i]]$, the type of $\Suf(j)$ can be obtained immediately, where $j=\SA[i]-1$.
Otherwise $T[j]= T[\SA[i]]$ (this case $\Suf(j)$ belongs to the current scanning bucket $T[\SA[i]]$), if all S-suffixes of $T$ that belong to bucket $T[\SA[i]]$ are not already sorted, then the $\Suf(j)$ is S-suffix.
\end{lemma}

It is not hard to decide whether all S-suffixes in the current scanning bucket $T[\SA[i]]$ are already sorted or not.
One can use an extra variable to denote how many S-suffixes remain in the current scanning bucket $T[\SA[i]]$.
When we begin to scan a new bucket, we scan this bucket once to initialize this variable.
Note that we can do this initialization, since there are special symbols in $\SA_S$ for each bucket after we initialize the buckets in our interior counter trick.

Now, there is only one thing left: how to know the tails of the bucket of S-suffixes in the induced sorting step.
The purpose of the pointer data structure is to indicate the tails of the bucket of S-suffixes.
However, the pointer data structure requires $c_p$ words, where the value of $c_p$ will be specified later.
Thus, we need to divide this induced-sorting step into two stages.
The first stage (Section~\ref{sec:stage1}) is to sort the first $n_S-c_p$ S-suffixes (i.e. the largest $n_S-c_p$ S-suffixes), where our pointer data structure still exists.
The second stage (Section~\ref{sec:stage2}) is to sort the last $c_p$ S-suffixes, where there is no space for the pointer data structure.

\subsubsection{The first stage}
\label{sec:stage1}
\topic{Pointer data structure}
Now, we construct our pointer data structure which supports to find the tails of the buckets in constant time.
We store the pointer data structure in the tail of $\SA_S$, recall that $\SA_S=\SA[n_L,\ldots,n-1]$.
Now, we describe the details.
We divide the S-suffixes of $T$ into $4d$ parts according to their first characters,
and construct the pointer data structure for each part respectively.
The $4d$ parts are divided by $T[j]\in\left[\frac{i|\Sigma|}{4d} +1,\frac{(i+1)|\Sigma|}{4d}\right]$,
for $0\leq i<4d$.
Let $D_i$ denote the pointer data structure of the $i$-th part.
We only show the details how we construct the pointer data structure $D_0$ as follows,
since constructing $D_i$ is similar for $0<i<4d$ (one only need to shift $T[j]$ with $\frac{i|\Sigma|}{4d}$).
\begin{enumerate}[(1)]
         \item First, we let $\SA_S[i]=1$ for all $i\in[1,\frac{|\Sigma|}{4d}]$.
         Then we scan $T$ from right to left.
         For every S-type $T[i]\in [1,\frac{|\Sigma|}{4d}]$,
         we increase $\SA_S[T[i]]$ by one.

         \item
         Then we scan $\SA_S[1 \ldots \frac{|\Sigma|}{4d}]$ from left to right.
         We use a variable $sum$ to count the sum, first initialize $sum =-1$.
         For each $\SA_S[i]$ which is being scanned, first let $sum = sum + \SA_S[i]$, then let $\SA_S[i]= sum$.
         Now, for any S-suffix $\Suf(i)$ satisfying $T[i]\in [1,\frac{|\Sigma|}{4d}]$, $\SA_S[T[i]]-T[i]$ must indicate the tail of bucket $T[i]$ in $\SA_S$.
         Since we want every entry in $\SA_S[1\ldots\frac{|\Sigma|}{4d}]$ to be distinct,
         we initialize $\SA_S[i]=1$ for all $i\in[1,\frac{|\Sigma|}{4d}]$ in Step (1).
         Hence the tail of bucket $T[i]$ is $\SA_S[T[i]]-T[i]$.

         \item
         Finally, we construct $D_0$ for $\SA_S[1\ldots\frac{|\Sigma|}{4d}]$ according to Lemma~\ref{lem:rointselect}.
         $D_0$ uses at most $c(n+\frac{|\Sigma|}{4d})/\log n$ words space.
         We store $D_0$ in the tail of $\SA_S$ (i.e., $\SA_S[n_S-c(n+\frac{|\Sigma|}{4d})/\log n\ldots n_S-1]$).
         $D_0$ supports to find the tail of the bucket of any S-suffix $\Suf(i)$ satisfying $T[i]\in [1,\frac{|\Sigma|}{4d}]$ in constant time.
       \end{enumerate}

\begin{lemma}\label{lem:rointselect}
	For any $m$ distinct integers $0\leq a_0<a_1\ldots<a_{m-1}\leq n$, where $m\leq n$ and $n>1024$,  one can construct a data structure using
	linear time (i.e., $O(n)$ time) and
	at most $cn/\log n$ words, where $1<c<2$,
	such that
	each query to the $i$-th smallest integer $a_i$ ($\select(i)$) can be answered in constant time.
\end{lemma}
\begin{proof}
We first construct a bitmap $B[0\ldots n]$.
we initialize $B$ by $B[a_i] = 1$ for all $i\in [0, m-1]$.
We need a data structure to
support query $\select(i)$, which asks for
the index of $i$-th $1$ in $B$.
There is an auxiliary data structure using $O(n/\log\log n)$ bits
(more precisely $3n/\log\log n + n^\frac{1}{4}(\frac{1}{4}\log n\log\log n+\log\log n)$)
which can be constructed in $O(n)$ time
to support constant time $\select$ query in $B$~\cite{jacobson1989space,clark1996compact,NPsea12.1}.
Converting bits to words, we can see that
the data structure uses at most $cn/\log n$ words
(for $1<c<2$ if $n>1024$).
\end{proof}

After this step, the pointer data structure (i.e. $D_i$ for all $0\le i<4d$) is stored in $\SA_S[n_S-c_p\ldots n_S-1]$,
where $c_p=\lceil 4d\cdot(c(n+\frac{|\Sigma|}{4d})/\log n) \rceil
\leq \lceil 5dcn/\log n \rceil$.
Recall that the sorted $n_L$ L-suffixes are stored in $\SA_L$ and $n_S \ge n_L$.
Thus the empty space (i.e., $\SA_S=\SA[n_L,\ldots,n-1]$) is at least half of the space of $\SA$.
It is enough to construct and store the pointer data structure, i.e.,
$\frac{|\Sigma|}{4d}\le \frac{dn}{4d}=\frac{n}{4}$ for constructing it and $\frac{n}{4}$ for storing it (which uses $c_p$ words).
Note that we assume that $n>4c_p$, otherwise it is easy to solve since $n$ is constant.
Hence the pointer data structure is constructed in linear time and supports to find the tail of the bucket of any S-suffix in constant time.
\begin{lemma}\label{lem:intpds}
We can construct the pointer data structure in linear time, and this pointer data structure uses at most $c_p$ words and
can support to find the bucket tail of any S-suffix in constant time.
\end{lemma}
\begin{proof}
We only need to specify the query time constant.
We use $4d$ values, i.e.,
$m_1,\ldots,m_{4d}$,
which denote the number of S-suffixes in each interval, respectively
(They can be obtained from the variable $sum$ which is computed in the final stage of the Step (2)).  Now, if we want to find the bucket tail of an S-suffix $\Suf(i)$, we first compare $T[i]$ with $\frac{i|\Sigma|}{4d}$ (for $0\le i<4d$) to see
which pointer data structure
$T[i]$ belongs to. Assume that it belongs to $D_j$.
Then we do a $\select(T[i]-(j-1)\frac{|\Sigma|}{4d})$ query on $D_j$, and combine the $\select$ result with the corresponding $m_k$ $(k<j)$ to identify the tail of bucket $T[i]$. All the above operations can be done in constant time.
\end{proof}

Now according to Lemma~\ref{lem:inttrick}, \ref{lem:idls} and \ref{lem:intpds}, we can sort the first largest $n_S-c_p$ S-suffixes from the sorted L-suffixes which stored in $\SA_L$ using the induced sorting step with the interior counter trick and the pointer data structure which stored in $\SA[n-c_p\ldots n-1]$, and we store the indices of the sorted $n_S-c_p$ S-suffixes in $\SA_S[0,\ldots,n_S-c_p-1]$.

\subsubsection{The second stage}
\label{sec:stage2}

Now, we only need to show how to sort the last $c_p$ S-suffixes which is occupied by our pointer data structure.
First, we specify that how to identify whether an S-suffix belongs to the first $n_S-c_p$ S-suffixes or not.
We can scan $T$ once to find the character of $n_S-c_p$ largest S-suffix using the pointer data structure and let $ch$ denote this character.
If the tail of the bucket $ch$ is exactly $c_p-1$ in $\SA_S$, then we compare the beginning character of the S-suffix with $ch$ to identify which part it belongs to.
Otherwise, the $n_S-c_p$ largest S-suffix belongs to bucket $ch$, and it will be stored in $\SA_S[c_p-1]$.
We only need two variables to indicate whether the S-suffix belongs to the first $n_S-c_p$ S-suffixes or not.
One is the number to denote the tail of bucket $ch$, and the other is also an integer number which denotes the gap between the tail of bucket $ch$ and $c_p-1$.

Now, we describe the details to sort the last $c_p$ S-suffixes.
First, we move these largest $n_S-c_p$ S-suffixes to the tail of $\SA_S$, i.e., $\SA_S[c_p,n_S-1]$.
Then we scan the $T$ from right to left to place the smallest $c_p$ S-suffixes into $\SA_S[0,c_p-1]$.
Now, we use merge sort with the in-place linear time merging algorithm~\cite{Salowe1987} to sort these $c_p$ S-suffixes, the sorting key for each S-suffix is its beginning character.
After this sorting step, these $c_p$ S-suffixes have been placed in their corresponding buckets in $\SA_S[0,c_p-1]$.
Note that we can use the same sorting step (which we used for sorting the first $n_S-c_p$ S-suffixes) to sort the last $c_p$ S-suffixes without the pointer data structure.
Now, the key point is that we can use the binary search (instead of the pointer data structure) to find the tails of the bucket for these $c_p$ S-suffixes,
since $c_p\leq \lceil 5dcn/\log n \rceil$ is small enough (i.e. $c_p\log n=O(n)$) to maintain that the time complexity of our algorithm is $O(n)$.

Using the binary search to extend interior counter trick is not very difficult, one can see the details in Section~\ref{sec:rogel4} which we induced sort all L-suffixes from the sorted S-suffixes for the general alphabets (Note that the optimal time is $O(n\log n)$ for the general alphabets case.
Thus, we directly use the binary search to extend interior counter trick and
do not use the pointer data structure in that case.).

After this step, all $n_S$ S-suffixes are sorted in $\SA_S$.
Now we have all sorted L-suffixes in $\SA_L$ (i.e., $\SA[0\ldots n_L-1]$) and all sorted S-suffixes in $\SA_S$ (i.e., $\SA[n_L\ldots n-1]$).
We use the stable, in-place, linear time merging algorithm~\cite{Salowe1987} to merge the ordered $\SA_L$ and $\SA_S$ (the merging key for $\SA[i]$ is $T[\SA[i]]$, i.e., the first character of \suf{\SA[i]}). After this merging step, all suffixes of $T$ have be sorted in $\SA[0\ldots n-1]$.

Finally, we obtain the following theorem for our optimal in-place algorithm.
\begin{theorem}[Main Theorem]
\label{thm:roint}
Our Algorithm takes $O(n)$ time and $O(1)$ workspace to compute the suffix array of string $T$ over integer alphabets $\Sigma$, where $T$ is read-only and $|\Sigma|=O(n)$.
\end{theorem}

\section{Suffix Sorting for Read-only General Alphabets}
\label{sec:gel}

\subsection{Framework}

The framework of our algorithm for read-only general alphabets (i.e., the only operations allowed on the characters of $T$ (read-only) are comparisons) is described as follows:
\begin{enumerate}
  \item If $n_S \leq n_L$ (i.e., the number of S-suffixes is no larger than that of L-suffixes), then
    \begin{enumerate}[(1)]
      \item (Section \ref{sec:rogel1}) Sort all S-substrings of $T$ using mergesort directly.\\
          We use mergesort to sort all S-substring of $T$ in $\SA[n-n_S\ldots n-1]$. In the merging step of mergesort, we use $\SA[0\ldots n_S-1]$ as the temporary space. After this step, all S-substrings should be in the lexicographical order stored in $\SA[n-n_S\ldots n-1]$.
      \item (Section \ref{sec:rogel2}) Construct the reduced problem $T_1$ from the sorted S-substrings.\\
          We construct the reduced problem $T_1$ using the ranks of all sorted S-substrings which are stored in $\SA[n-n_S\ldots n-1]$. The ranks of S-substrings are corresponding to the lexicographical order of the sorted S-substrings.  After this step, we get the reduced problem $T_1$ in $\SA[0\ldots n_S-1]$.
      \item (Section \ref{sec:rogel3}) Sort the S-suffixes by solving $T_1$ recursively.\\
          We sort $T_1=\SA[0\ldots n_S-1]$ recursively.
          In the recursive step, we use $\SA_1=\SA[n-n_S\ldots n-1]$ as the output space for $T_1$.
          Then we use the suffix array of $T_1$ (i.e. $\SA_1$) to place all indices of the sorted S-suffixes of $T$ into $\SA[n-n_S\ldots n-1]$.
      \item (Section \ref{sec:rogel4}) Induced sort all suffixes from the sorted S-suffixes.\\
          First, we place all indices of S-suffixes in their final positions in $\SA$ by using mergesort together with a stable, in-place, linear time merging algorithm~\cite{Salowe1987}. Then we extend our \emph{interior counter trick} to sort all L-suffixes from the sorted S-suffixes. Finally, all indices of the sorted suffixes of $T$ are stored in $\SA[0\ldots n-1]$.
    \end{enumerate}
  \item Otherwise, execute the above steps switching the roles of $L$ and $S$.
\end{enumerate}

Note that our algorithm is very simple since the optimal time complexity in this case is $O(n\log n)$ instead of $O(n)$.
Moreover, our algorithm does not make any bit operations rather than previous algorithms, e.g., \cite{franceschini2007place}.

Without loss of generality, we assume that $n_S\leq n_L$.
Now, we describe the details of our in-place algorithm in the following sections.

\eat{
\begin{algorithm}[!htb]
\caption{Optimal in-place suffix sorting algorithm for read-only general alphabets}
\label{alg:3}
\begin{algorithmic}[1]
\REQUIRE $T$;
\IF{$n_S\leq n_L$}
\STATE Use mergesort to sort all S-substrings of $T$ and store them in $\SA[n-n_S\ldots n-1]$, using $\SA[0\ldots n/2]$ as the temporary space;
\STATE Construct the reduced recursive problem $T_1$ using the ranks of the sorted S-substrings and store $T_1$ in $\SA[0\ldots n_S-1]$, using naive comparison to obtain the ranks;
\STATE Let $\SA_1$ denote $\SA[n-n_S\ldots n-1]$ which is reused as the output space for problem $T_1$;
\IF{all characters in $T_1$ are unique}
\STATE Directly construct $\SA_1$ for $T_1$;
\ELSE \STATE Obtain $\SA_1$ by solving $T_1$ recursively, reusing the space of $\SA$ by induction;
\ENDIF
\STATE Sort all S-suffixes of $T$ from $\SA_1$ and store them in $\SA[n-n_S\ldots n-1]$, simply replacing $\SA_1$ with the indices of S-suffixes;
\STATE Induced sort all suffixes of $T$ from the sorted S-suffixes, using our extended interior counter trick;
\ELSE \STATE Execute above steps (i.e., Line 2--11) switching the role of S with L;
\ENDIF
\RETURN $\SA$
\end{algorithmic}
\end{algorithm}
}

\subsection{Sort all S-substrings of $T$}
\label{sec:rogel1}
In this section, we sort all S-substrings of $T$
as follows:

\begin{enumerate}
  \item First, we scan $T$ from right to left
  and place all indices of S-type characters into $\SA[n-n_S\ldots n-1]$. Note that $n_S \leq n/2$ since we assume that $n_S \leq n_L$.
  \item Then, we sort $\SA[n-n_S\ldots n-1]$ using mergesort (the sorting key for $\SA[i]$ is the S-substring of $T$ which begins at $T[\SA[i]]$). We use $\SA[0\ldots n_S-1]$ as the temporary space for mergesort.
  To compare two keys (i.e., two S-substrings) in mergesort, we simply do the straightforward character-wise comparisons.
\end{enumerate}
After the above two steps, all the S-substring have been sorted in $\SA[n-n_S\ldots n-1]$. We have the following lemma.

\begin{lemma}
	\label{lm:length}
	We can sort all S-substrings using $O(n\log n)$ time and $O(1)$ workspace.
\end{lemma}

\begin{proof}
Step 1 does not need any extra space and costs linear time, because we can compute the type of each character in $O(1)$ time during the right-to-left scan of $T$.
Moreover, we know mergesort needs linear workspace.
Hence, it is sufficient to use $\SA[0\ldots n_S-1]$ as the workspace for mergesort.
For Step 2, it suffices to show that the time spent for comparison process in any recursive level of mergesort (there are $O(\log n)$ recursive levels) can be bounded by $O(n)$. In any level, each S-substring is compared to exactly one other S-substring. The length of the S-substrings can be obtained according to Observation \ref{ob:nexts}. Recall that each character of $T$ is scanned at most twice since it only be scanned when identifying the length of its adjacent predecessor S-substring and itself. Thus the comparison process takes $O(n)$ time in any level because the total length of all S-substrings is less than $2n$.
\end{proof}

\subsection{Construct the reduced problem $T_1$ from the sorted S-substrings}
\label{sec:rogel2}
In this section, we construct the reduced problem $T_1$ by renaming the sorted S-substrings. After Section \ref{sec:rogel1}, all S-substrings have been sorted in
$\SA[n-n_S\ldots n-1]$.
The construction of $T_1$ consists of the following
two steps:
\begin{enumerate}
  \item
  We rename the S-substrings by their ranks.
  First let the rank of $\SA[n-n_S]$ be 0. We scan $\SA[n-n_S+1\ldots n-1]$ from left to right. When scanning $\SA[i]$,
  we compare S-substring beginning with $T[\SA[i]]$ and S-substring beginning with $T[\SA[i-1]]$.
  If they are different, let the rank of $\SA[i]$ be the rank of $\SA[i-1]$ plus one. Otherwise, the rank of $\SA[i]$ is the same as that of $\SA[i-1]$.
  We store the rank of $\SA[i]$
  in $\SA[i-n+n_S]$.
  \item Next, we use the heapsort to sort $\SA[n-n_S\ldots n-1]$
  (the sorting key for $\SA[i]$ is $\SA[i]$ itself).
  When we exchange two entries
  (say, $\SA[i]$ and $\SA[j]$, $i,j \in[n-n_S\ldots n-1]$)
  in $\SA[n-n_S\ldots n-1]$
  during heapsort,
  we also exchange the corresponding two entries
  (i.e., $\SA[i-n+n_S]$ and $\SA[j-n+n_S]$)
  in $\SA[0\ldots n_S-1]$.
  Note that we use heapsort here since it is in-place, so we do not need any extra space.
\end{enumerate}

After the above two steps, we get the reduced problem $T_1$ in $\SA[0\ldots n_S-1]$.

\begin{lemma}
$T_1$ can be constructed in $O(n\log n)$ time and $O(1)$ workspace.
\end{lemma}
\begin{proof}
In Step 1, each S-substring beginning with $T[\SA[i]]$ is compared with S-substring beginning with $T[\SA[i+1]]$.
So each S-substring only participates in two comparisons. Now the argument is similar to a comparison process in a recursive level of mergesort in Lemma \ref{lm:length}, thus it costs linear time.
Obviously, Step 2 takes $O(n\log n)$ time and $O(1)$ workspace.
\end{proof}

\subsection{Sort the S-suffixes by solving $T_1$ recursively}
\label{sec:rogel3}
In this section, we solve $T_1=\SA[0\ldots n_S-1]$ recursively to obtain
the order of all S-suffixes in $\SA[n-n_S\ldots n-1]$. For the recursive step, we use $\SA_1=\SA[n-n_S\ldots n-1]$ as the output space for $T_1$.
After the recursive call, $\SA_1$ stores the suffix array of $T_1$.
We need to restore their names back to the indices of S-suffixes in $T$ they represented.
This step can be done as follows.

\begin{enumerate}
  \item First, we scan $T$ from right to left.
  We maintain a counter $sum$ for the number of S-type characters we have scanned so far.
  Initially $sum$ is 0. If $T[i]$ is S-type, we increase $sum$ by 1 and
  place $\Suf(i)$ into $\SA[n_S-sum]$ (i.e., let $\SA[n_S-sum]\leftarrow i$). Now $\SA[0\ldots n_S-1]$ stores the indices of all S-suffixes of $T$.
  \item Then for $i\in[n-n_S,n-1]$, let $\SA[i]\leftarrow \SA[\SA[i]]$.
\end{enumerate}

Now, we have obtained all S-suffixes in the lexicographical order in $\SA[n-n_S\ldots n-1]$.

\subsection{Induced sort all suffixes of $T$}
\label{sec:rogel4}

From Section \ref{sec:rogel3}, we have obtained the sorted S-suffixes in $\SA[n-n_S\ldots n-1]$. Now, we sort all suffixes from these sorted S-suffixes.

\topic{Preprocessing}
First, we scan $T$ from right to left to place all indices of L-suffixes into $\SA[0\ldots n-n_S-1]$.
Then, we sort $\SA[0\ldots n-1]$ (the sorting key of $\SA[i]$ is $T[\SA[i]]$ i.e., the first character of \suf{\SA[i]}) using the mergesort, with
the merging step implemented by the stable, in-place, linear time merging algorithm developed by Salowe and Steiger~\cite{Salowe1987}.
After this sorting step, we show some useful observations.

\begin{observation}\label{ob:allsuf}
All suffixes of $T$ have been sorted by their first characters in $\SA$, i.e., in their corresponding buckets.
\end{observation}

\begin{observation}\label{ob:lsuf}
All indices of L-suffixes beginning with the same character in $\SA$ are in increasing order, due to the stableness of the above sorting algorithm.
\end{observation}

\begin{lemma}\label{lem:gels}
All S-suffixes are already in their final position in $\SA$.
\end{lemma}
\begin{proof}
Before the sorting step, all sorted S-suffixes are in $\SA[n-n_S\ldots n-1]$ and all L-suffixes are in $\SA[0\ldots n-n_S-1]$. Because the merging step is stable, the S-suffixes are behind the L-suffixes in the same bucket and hence are already in their final positions in $\SA$ from Observation \ref{ob:allsuf} and Property~\ref{prop:ls}.
\end{proof}

\topic{Induced Sorting}
Now, we induce the order of all L-suffixes from the sorted S-suffixes (which are already in their final position in $\SA$ by Lemma \ref{lem:gels}) using induced sorting.
Now, we extend the interior counter trick in Section~\ref{sec:roint2}
to handle the read-only general alphabets.
We use five special symbols $\BH$ (Head of L-suffixes), $\BT$ (Tail of L-suffixes), $\Ept$ (Empty), $\Rone$ (one remaining L-suffix) and $\Rtwo$ (two remaining L-suffixes)
\footnote{
	Recall that the special symbol is only used to simplify the argument. See Appendix \ref{app:special} for the details. The only difference is that we need to use the in-place linear time merging algorithm~\cite{Salowe1987} to place all indices of the suffixes of $T$ in their corresponding buckets in $\SA$ (the merging key is the beginning character of the suffix), then we scan the $\SA$ once to know the bucket heads/tails for the five integers (instead of scanning the string $T$ for the (read-only) integer alphabets).
	}.
We do the following two steps to sort all L-suffixes:

\topic{Step 1. Initializing $\SA$}
Firstly, we initialize all buckets in $\SA$ by placing some special symbols in each bucket in order to inform us the number of L-suffixes in the bucket.
Concretely, we scan $T$ from right to left.
For each scanning character $T[i]$ which is L-type,
if bucket $T[i]$ has not been initialized,
we need to initialize bucket $T[i]$ (we will show how to identify the bucket is initialized or not in the end of this step). Before to initialize bucket $T[i]$, we first need to obtain the value $N_L$, which is the number of L-suffixes in this bucket.
Let $l$ denote the head of bucket $T[i]$ in $\SA$ (i.e. $l$ is the smallest index in $\SA$ such that $T[\SA[l]]=T[i]$) and $r$ denote the tail of bucket $T[i]$ in $\SA$ (i.e. $r$ is the largest index in $\SA$ such that $T[\SA[r]]=T[i]$).
Furthermore, we let $r_L$ denote the tail of L-suffixes in this bucket (i.e., $r$ is the largest index in $\SA$ such that $T[\SA[r_L]]=T[i]$ and $T[\SA[r]]$ is L-type).
Note that $N_L=r_L-l+1$.
Hence, it suffices to compute $l$ and $r_L$.
The following steps compute $l$ and $r_L$, respectively.
         \begin{enumerate}[(i)]
            \item We can find $l$ by searching $T[i]$ in $\SA$ (the search key for $\SA[i]$ is $T[\SA[i]]$) using \emph{binary search}. This uses $O(\log n)$ time from Observation \ref{ob:allsuf}.
            \item For $r_L$, since the bucket $T[i]$ has not been initialized, $\Suf(i)$ is the first L-suffix in its bucket being scanned. From Observation \ref{ob:lsuf}, $\Suf(i)$ must be stored in $\SA[r_L]$ (i.e., $\SA[r_L]=i$) since we scan $T$ from right to left.
                Hence, we can scan this bucket from $l$ to $r$ to find $r_L$ which satisfies $\SA[r_L]=i$.
          \end{enumerate}
After this, we have obtained the value of $N_L$. Now, we initialize the bucket $T[i]$ as follows:
\begin{enumerate}[(1)]
\item If $N_L = 1$, we do nothing (there is only one L-suffix in this bucket, and obviously it is in the final position).
\item If $N_L = 2$, let $\SA[l+1] = \BT$ (recall that $l$ is the head of bucket $T[i]$ and $r$ is the bucket tail, i.e., $\SA[l\ldots r]$ is the bucket $T[i]$. Moreover, $r_L$ is the tail of L-suffixes in this bucket)
\[
\begin{array}{cccccc}
\mathsf{Index}     & l & l+1(r_L) & l+2 &$\ldots$& r \\
\mathsf{Type}       & L &  L  & S& $\ldots$& S   \\
\SA       &\SA[l]  & \urb{\BT}  & \SA[l+2]  &$\ldots$& \SA[r]
\end{array}
\]
  \item If $N_L = 3$, let $\SA[l+1]=\BH$ and $\SA[l+2]=\BT$.
\[
\begin{array}{ccccccc}
\mathsf{Index}  & l & l+1 & l+2(r_L) &l+3&$\ldots$& r \\
\mathsf{Type}       & L &  L  & L  & S& $\ldots$& S   \\
\SA     & \SA[l] & \urb{\BH} & \urb{\BT}  & \SA[l+3] &$\ldots$& \SA[r]
\end{array}
\]
  \item If $N_L>3$, let $\SA[l+1]=\BH$, $\SA[l+2]=\Ept$ and $\SA[l+N_L-1]=\BT$.
\[
\begin{array}{cccccccccc}
\mathsf{Index}    & l & l+1 & l+2 & l+3 & $\ldots$  & $$l+N_L-1(r_L)$$   &$$l+N_L$$ & $$\ldots$$ & r\\
\mathsf{Type}     & L &  L  & L   & L &  $\ldots$ & L &  S& $\ldots$& S \\
\SA     & \SA[l] &  \urb{\BH}  & \urb{\Ept}   & \SA[l+3] & $\ldots$ &\urb{\BT} & \SA[l+N_L] &$\ldots$& \SA[r]
\end{array}
\]
\end{enumerate}
Note that we can find out whether the bucket $T[i]$ is already initialized or not in $O(\log n)$ time. We do a binary search find $l$, then check $\SA[l+1]$ is $\BH,\BT$ or others. If the bucket $T[i]$ has been initialized, $\SA[l+1]$ is $\BH$ or $\BT$. Otherwise, it has not been initialized yet
\footnote{
	Note that we are able to do the binary search in $\SA$, though there are special symbols (i.e.~$\BH, \BT,\Ept$) in $\SA$. Since the longest continuous special symbol entries in $\SA$ is $2$, i.e., any three of continuous entries in $\SA$ must have at least one suffix entry (this entry represents the index of a suffix of $T$).
}.
It is not hard to see that this initialization step uses $O(n\log n)$ time and $O(1)$ workspace.

\topic{Step 2. Sort all L-suffixes using induced sorting}
We scan $\SA$ from left to right to sort all L-suffixes. The step is similar to sorting all suffixes in Section \ref{sec:roint2}. The main difference is that we use binary search to find the bucket head. Specifically, we scan $\SA$ from left to right.
For every $\SA[i]$, let $j= \SA[i]-1$.
If $T[j]$ is L-type, then place \suf{j} into the LF-entry of its bucket,
and increase the head counter by one.
To specify how to place the L-suffix into the LF-entry of its bucket in $\SA$,
We only specify the case where $N_L>3$ for the bucket. The other cases with $N_L\leq 3$ are similar and simpler. Let $L_i$ denote the index of $i$-th L-suffix which needs to be placed into the LF-entry of this bucket. We distinguish the following four cases:

\begin{enumerate}[(1)]
    \item If $\SA[l+1]=\BH$ and $\SA[l+2]=\Ept$: The first L-suffix (i.e.~$L_1$) needs to be placed into this bucket. \\
         We let $\SA[l]=j$ and $\SA[l+2]=1$ (use $\SA[l+2]$ as the counter to denote the number of L-suffixes have been placed so far).  Recall that \suf{j} is the current L-suffix we want to place.
    \item If $\SA[l+1]=\BH$ and $\SA[l+2]\neq \Ept$: These L-suffixes except the first L-suffix ($L_1$) and the last two L-suffixes ($L_{N_L-1}$ and $L_{N_L}$) need to be placed.\\
     Let $c=\SA[l+2]$ (counter).  If $\SA[l+c+2] \neq \BT$, we let $\SA[l+c+2]=j$ and $\SA[l+2]=c+1$. Otherwise (this is the third to last L-suffix, i.e.~$L_{N_L-2}$), we shift these $c-1$ L-suffixes to the left by one position (i.e., move $\SA[l+3\ldots r_L-1]$ to $\SA[l+2\ldots r_L-2]$ ) and let $\SA[r_L-1]=j$ and $\SA[l+1]=\Rtwo$.
    \item If $\SA[l+1]=\Rtwo$: The penultimate L-suffix (i.e.~$L_{N_L-1}$) needs to be placed. \\
     We scan this bucket to find $r_L$ such that $\SA[r_L]=\BT$. Then, we move $\SA[l+2\ldots r_L-1]$ to $\SA[l+1\ldots r_L-2]$. After, we let $\SA[r_L-1]=j$ and $\SA[r_L]=\Rone$.
    \item Otherwise, the last L-suffix (i.e.~$L_{N_L}$) needs to be placed. \\
     We scan this bucket to find $r_L$ such that $\SA[r_L]=\Rone$. Then let $\SA[r_L]=j$.
  \end{enumerate}
\[
\begin{array}{cccccccccc}
\mathsf{Index}    & l & l+1 & l+2 & l+3 & $\ldots$  & $$l+N_L-1(r_L)$$   &$$l+N_L$$ & $$\ldots$$ & r\\
\mathsf{Type}     & L &  L  & L   & L &  $\ldots$ & L &  S& $\ldots$& S \\
\mathsf{Case}~(1):\\
\SA     & \SA[l] &  \BH  & \Ept   & \SA[l+3] & $\ldots$ &\BT & \SA[l+N_L] &$\ldots$& \SA[r] \\
\SA     & \urb{L_1}& \BH & \urb{1}   & \SA[l+3] & $\ldots$ &\BT & \SA[l+N_L] &$\ldots$& \SA[r] \\
\mathsf{Case}~(2):\\
\SA     & L_1& \BH & 1   & \SA[l+3] & $\ldots$ &\BT & \SA[l+N_L] &$\ldots$& \SA[r] \\
\SA     & L_1& \BH & \urb{2}   & \urb{L_2} & $\ldots$ &\BT & \SA[l+N_L] &$\ldots$& \SA[r] \\
\vdots           &&&\vdots&&&&&& \\
\SA     & L_1& \BH & \urb{N_L-3}   & L_2 & $\ldots$ &\BT & \SA[l+N_L] &$\ldots$& \SA[r] \\
\SA     & L_1& \urb{\Rtwo} & \urb{L_2}   & \urb{L_3} & $\ldots$ &\BT & \SA[l+N_L] &$\ldots$& \SA[r] \\
\mathsf{Case}~(3):\\
\SA     & L_1& \Rtwo & L_2   & L_3 & $\ldots$ &\BT & \SA[l+N_L] &$\ldots$& \SA[r] \\
\SA     & L_1& \urb{L_2} & \urb{L_3}   & \urb{L_4} & $\ldots$ &\urb{\Rone} & \SA[l+N_L] &$\ldots$& \SA[r] \\
\mathsf{Case}~(4):\\
\SA     & L_1& L_2 & L_3   & L_4 & $\ldots$ &\Rone & \SA[l+N_L] &$\ldots$& \SA[r] \\
\SA     & L_1& L_2 & L_3   & L_4 & $\ldots$ &\urb{L_{N_L}} & \SA[l+N_L] &$\ldots$& \SA[r] \\
\end{array}
\]

We have the following lemmas and theorem.
\begin{lemma}\label{lem:scan}
When we scan $\SA[i]$ in the induced sorting step, we can identify whether $T[\SA[i]-1]$ is L-type or S-type in $O(1)$ time. The only exception is when $\Suf(\SA[i]-1)$ is the last L-suffix which needs to be inserted into the bucket $T[\SA[i]-1]$. This special case requires $O(N_{L})$ time, where $N_{L}$ denotes the number of L-suffixes in its bucket (i.e.~$N_{L}=r_L-l+1$).
\end{lemma}
\begin{proof}
Let $j = \SA[i]-1$. First if $T[j] \neq T[j+1]$, by definition it is trivial. Otherwise, $T[j] = T[j+1]$. From Lemma \ref{lem:idls}, we only need to know whether all L-suffixes in the bucket $T[j]$ (i.e., bucket $T[\SA[i]]$) have already been sorted or not. In our algorithm, we use the extended interior counter trick which maintains the counters of the buckets. So we can identify whether all L-suffixes in the bucket $T[j]$ have already been sorted or not immediately except when $\Suf(j)$ is the last L-suffixes which needs to be placed into the bucket $T[j]$ (corresponding to case (4)). However, we can scan this bucket from left to right to identify whether the special symbol $R_1$ exists or not. If exists, which means there is one L-suffix remained, this must be the $\Suf(j)$. Otherwise, all L-suffixes in the bucket $T[j]$ have already been sorted. This scanning operation takes $O(N_{L})$ time.
\end{proof}

\begin{lemma}
All the suffixes can be sorted correctly from the sorted S-suffixes in $O(n\log n)$ time.
\end{lemma}
\vspace{-1.5mm}
\begin{proof}
All S-suffixes have been sorted correctly according to Lemma \ref{lem:gels}.  All L-suffixes can be sorted correctly from the sorted S-suffixes using induced sorting steps according to Lemma \ref{lem:ls}. For the time complexity, in the scanning process, we use $O(\log n)$ time to do binary search to find the head of its bucket for each L-suffix (identified by Lemma \ref{lem:scan}), and use the extended  interior counter trick to find the LF-entry in constant time except for the last two L-suffixes of the bucket. We need to scan all L-suffixes in this bucket in order to find the final position (corresponding to case (4) in Step 2) and shift these L-suffixes by one position (corresponding to case (3) in Step 2). This only takes $O(n)$ time overall since each bucket is scanned at most twice. To sum up, this sorting step takes $O(n\log n)$ time.
\end{proof}

\vspace{-1mm}
\begin{theorem}
Our Algorithm takes $O(n\log n)$ time and $O(1)$  workspace to compute the suffix array for string $T$ over general alphabets, where $T$ is read-only and only comparisons between characters are allowed.
\end{theorem}
\vspace{-1.5mm}
\begin{proof}
All steps in our algorithm take $O(n\log n)$ time. Besides, the size of recursive problem $T_1$ is no larger than half of $|T|$. We have $T(n) = T(n/2)+ O(n\log n)$, thus $T(n) =O(n\log n +\frac{n}{2}\log\frac{n}{2}+\cdots) = O(n\log n)$. For workspace, every step uses $O(1)$ workspace, and in the recursive subproblem we can also reuse the $O(1)$ workspace. Moreover, at the same recursive level, the different steps can reuse the $O(1)$ workspace too.
\end{proof}

\section{Conclusion}
\label{sec:con}

In this paper, we present optimal in-place algorithms for suffix sorting over (read-only) integer alphabets and read-only general alphabets.
Concretely, we provide the first optimal linear time in-place suffix sorting algorithm for (read-only) integer alphabets. Our algorithms solve the open problem posed by Franceschini and Muthukrishnan in ICALP 2007~\cite{franceschini2007place}.
For the read-only general alphabets, we provide an optimal in-place $O(n\log n)$ time suffix sorting algorithm, which recovers the result obtained by Franceschini and Muthukrishnan \cite{franceschini2007place}. Our algorithm is arguably simpler.

There is a surge of interests in developing external memory algorithms for suffix sorting in recent years
(see e.g.,~\cite{Ferragina2012,Nong2015}).
Many such algorithms are extensions of existing lightweight internal memory algorithms.
It would be interesting to investigate the external memory setting and see whether our tricks and data structures are applicable in this setting.
We also plan to consider other string processing problems that are tightly connected with suffix array, such as compressed suffix arrays, longest common prefixes, Burrows-Wheeler transform and Lempel-Ziv factorization.


\section*{Acknowledgments}
The authors would like to thank Ge Nong and Gonzalo Navarro for helpful suggestions.


\newpage

\bibliographystyle{alpha}
\bibliography{sabib}

\newpage

\appendix
\renewcommand{\appendixname}{Appendix~\Alph{section}}

\section{Running Examples}
\label{app:eg}

\subsection{Induced sorting all L-suffixes from the sorted S-suffixes}
\label{app:inducedsort}

In this appendix, we give a running example for the standard induced sorting step which needs the bucket array and type array explicitly.
Assume that all indices of the sorted S-suffixes are already in their correct positions in
$\SA$ (i.e., in the tail of their corresponding buckets in $\SA$).
We scan $\SA$ from left to right (i.e., from $\SA[0]$ to $\SA[n-1]$).
When we scan $\SA[i]$, let $j=\SA[i]-1$.
If $T[j]$ is L-type, i.e., $\Suf(j)$ is an L-suffix (indicated by the type array),
we place the index of $\Suf(j)$ (i.e. $j$) into the LF-entry of bucket $T[j]$,
and then let the LF-pointer of this bucket $T[j]$ point to the next entry.
The LF-pointers are maintained by the bucket array.
If $\Suf(j)$ is an S-suffix, we do nothing (since all S-suffixes are already sorted in the correct positions).

The idea of induced sorting is that the lexicographical order between $\Suf(i)$ and $\Suf(j)$ are decided by the order of $\Suf(i+1)$ and $\Suf(j+1)$ if $\Suf(i)$ and $\Suf(j)$ are in the same bucket (i.e., $T[i]=T[j]$).
Considering two L-suffixes $\Suf(i)$ and $\Suf(j)$ in the same bucket, we have $\Suf(i+1)<\Suf(i)$ and $\Suf(j+1)<\Suf(j)$ by the definition of L-suffix. Since we scan $\SA$ from left to right, $\Suf(i+1)$ and $\Suf(j+1)$ must appear earlier than $\Suf(i)$ and $\Suf(j)$. Hence the correctness of induced sorting is not hard to prove by induction.
Note that the order of $\Suf(i)$ and $\Suf(j)$ with $T[i]\neq T[j]$ is already correct, since we always place the L-suffixes in their corresponding buckets.

\topic{Example}
We use the following running example to demonstrate the induced sorting step.
Consider a string $T[0\ldots 12]=``2113311331210"$ (the integer alphabets).
\[
\begin{array}{cccccccccccccc}
\mathsf{Index}  & 0 & 1 & 2 & 3 & 4 & 5 & 6 & 7 & 8 & 9 & 10& 11& 12 \\
  T    & 2 & 1 & 1 & 3 & 3 & 1 & 1 & 3 & 3 & 1 & 2 & 1 & 0 \\
\mathsf{Type}   & L & S & S & L & L & S & S & L & L & S & L & L & S
\end{array}
\]
$T[2]$ is S-type since $T[2]=1<T[3]=3$. The S-substrings are $\{11,1331,11,1331,1210,0\}$.
The S-suffixes are $\{\Suf(1),\Suf(2),\Suf(5),\Suf(6),\Suf(9),\Suf(12)\}$.

Now, we show the induced sorting step in our running example.
Suppose all indices of the sorted S-suffixes (i.e., $1,2,5,6,9,12$) are already stored in the tail of their corresponding buckets in $\SA$:
($\Ept$ denotes an Empty entry.)
\[
\begin{array}{cccccccccccccc}
\mathsf{Index}  & 0 & 1 & 2 & 3 & 4 & 5 & 6 & 7 & 8 & 9 & 10& 11& 12 \\
  T    & 2 & 1 & 1 & 3 & 3 & 1 & 1 & 3 & 3 & 1 & 2 & 1 & 0 \\
\mathsf{Type}  & L & S & S & L & L & S & S & L & L & S & L & L & S \\
  \SA   &(\urb{12})&(\Ept& \urb{1} & \urb{5} & \urb{9} & \urb{2}& \urb{6})&(\Ept& \Ept)&(\Ept & \Ept & \Ept & \Ept) \\
\mathsf{Bucket} &(0)&(1 & 1 & 1 & 1 & 1 & 1)&(2 & 2)&(3 & 3 & 3 & 3)
\end{array}
\]
(The entries between a pair of parentheses denote a bucket in $\SA$ which are these suffixes that start with the same character. The heads of bucket $0,1,2,3$ are $0,1,7,9,$ respectively.)

The scanning process is as follows. An arrow on top of a number indicates
that it is the current entry we are scanning.
\[
\begin{array}{cccccccccccccc}
\mathsf{Index}  & 0 & 1 & 2 & 3 & 4 & 5 & 6 & 7 & 8 & 9 & 10& 11& 12 \\
\mathsf{Type}   & L & S & S & L & L & S & S & L & L & S & L & L & S \\
  \SA   &(\vrb{12})&(\urb{11}& 1 & 5 & 9 & 2& 6)&(\Ept& \Ept)&(\Ept & \Ept & \Ept & \Ept) \\
  \SA   &(12)&(\vrb{11}& 1 & 5 & 9 & 2 &6)&(\urb{10}& \Ept)&(\Ept & \Ept & \Ept & \Ept) \\
  \SA   &(12)&( 11 &\vrb{1}& 5 & 9 & 2 &6)&( 10& \urb{0})&(\Ept & \Ept & \Ept & \Ept) \\
  \SA   &(12)&( 11 &1& \vrb{5} & 9 & 2 &6)&( 10& 0)&(\urb{4}& \Ept & \Ept & \Ept) \\
  \SA   &(12)&( 11 &1& 5 & \vrb{9} & 2 &6)&( 10& 0)&(4& \urb{8} & \Ept & \Ept) \\
  \SA   &(12)&( 11 &1& 5 &  9 & 2 & 6)&( 10& 0)&(\vrb{4}& 8 & \urb{3} & \Ept) \\
  \SA   &(12)&( 11 &1& 5 &  9 & 2 & 6)&( 10& 0)&( 4& \vrb{8} & 3 & \urb{7}) \\
\end{array}
\]
We first scan $\SA[0]=12$. We place $11$ (since $T[11]$ is L-type) to the LF-entry of bucket $1$ (i.e., $\SA[1]$), note that the LF-pointer of bucket $1$ initially points to $\SA[1]$ (head of bucket $1$).
Next, we scan $\SA[1]=11$, and we place $10$ ($T[10]$ is also L-type) to the LF-entry of its bucket (i.e., bucket $2$), and so on.\QEDB

\subsection{Induced sorting all L-suffixes from the sorted LMS-suffixes}
\label{app:lfromlms}
In this appendix, we show the induce sorting step which sorting the L-suffixes from LMS-suffixes in our example.
Note that the empty entries can be ignored, and all L-suffixes can still be sorted correctly.

\topic{Example}
Suppose all LMS-suffixes (i.e., $1,5,9,12$) are already sorted in the tail of their corresponding bucket in $\SA$:
($\Ept$ denotes an Empty entry.)
\[
\begin{array}{cccccccccccccc}
\mathsf{Index} & 0 & 1 & 2 & 3 & 4 & 5 & 6 & 7 & 8 & 9 & 10& 11& 12 \\
  T    & 2 & 1 & 1 & 3 & 3 & 1 & 1 & 3 & 3 & 1 & 2 & 1 & 0 \\
\mathsf{Type}  & L & S & S & L & L & S & S & L & L & S & L & L & S \\
  \SA   &(\urb{12})&(\Ept& \Ept & \Ept & \urb{1} & \urb{5}& \urb{9})&(\Ept& \Ept)&(\Ept& \Ept & \Ept& \Ept) \\
\mathsf{Bucket} &(0)&(1 & 1 & 1 & 1 & 1 & 1)&(2 & 2)&(3 & 3 & 3 & 3)
\end{array}
\]
The scanning process is as follows. An arrow on top of a number indicates
that it is the current entry we are scanning.
When we are scanning an empty entry in $\SA$, we ignore this entry (i.e., do nothing).
\[
\begin{array}{cccccccccccccc}
\mathsf{Index}  & 0 & 1 & 2 & 3 & 4 & 5 & 6 & 7 & 8 & 9 & 10& 11& 12 \\
\mathsf{Type}  & L & S & S & L & L & S & S & L & L & S & L & L & S \\
  \SA   &(\vrb{12})&(\urb{11}& \Ept& \Ept & 1 &5& 9)&(\Ept& \Ept)&(\Ept& \Ept & \Ept& \Ept) \\
  \SA   &(12)&(\vrb{11}& \Ept& \Ept & 1 & 5 &9)&(\urb{10}& \Ept)&(\Ept& \Ept & \Ept& \Ept) \\
  \SA   &(12)&( 11 & \Ept& \Ept &\vrb{1}& 5 & 9)&( 10& \urb{0})&(\Ept& \Ept & \Ept& \Ept) \\
  \SA   &(12)&( 11 & \Ept& \Ept& 1& \vrb{5} & 9 )&( 10& 0)&(\urb{4}& \Ept& \Ept & \Ept) \\
  \SA   &(12)&( 11 & \Ept& \Ept & 1& 5 & \vrb{9})&( 10& 0)&(4& \urb{8} & \Ept& \Ept) \\
  \SA   &(12)&( 11 &\Ept& \Ept &  1 & 5 & 9)&( 10& 0)&(\vrb{4}& 8 & \urb{3} & \Ept) \\
  \SA   &(12)&( 11 &\Ept& \Ept &  1 & 5 & 9)&( 10& 0)&( 4& \vrb{8} & 3 & \urb{7}) \\
\end{array}
\]\QEDB

\section{Handling Special Symbols}
\label{app:special}
We use at most five special symbols (e.g. U ($\Unique$), E ($\Empty$)) in this paper.
The special symbols are only used to simplify the argument, and we do not need any additional assumption.
Now we describe how to handle these special symbols.
Note that we introduce these special symbols and use them in our interior counter trick
(see Section \ref{secint:putlms}, \ref{secint:sortt} and \ref{sec:roint2}).
Recall that, if one uses the bucket array (which needs extra $\max\{|\Sigma|,n/2\}$ words workspace)
without using the interior counter trick, then the suffix array $\SA$ would only contain
the indices of $T$ (i.e., $\{0,1,\ldots,n-1\}$) (see the preliminary section).
Now the key point is to distinguish whether an entry in $\SA$ is one of these five integers (chosen to replace the special symbols) or just an index as before.

First, consider the simpler case where $n< 2^{\lceil\log n\rceil} -5$.
We can simply use integers $\{n,n+1,\ldots,n+4\}$ as these five special symbols
since they are different from all indices (identifiable) and each of them can be stored in one entry of $\SA$ same as an index of $T$.

Otherwise, we use any five integers (belong to $[0,n-1]$) as these five special symbols.
Without loss of generality, we assume that these five integers are $\{n-5,n-4,\ldots,n-1\}$.
Then we use ten extra variables (as the previous bucket array) to indicate the head/tail of the five buckets
(which the five suffixes $\{\Suf(n-5),\ldots,\Suf(n-1)\}$ belong to)
in $\SA$ and their LF/RF-entries. Thus we do not need the interior counter trick for these five buckets.
Note that we can obtain the head/tail of these five buckets by scanning $T$ once to count how many characters are smaller/larger than these five characters $\{T[n-5],\ldots,T[n-1]\}$, respectively.
To identify an entry of $\SA$, if the entry belongs to one of these five buckets,
it is just an index.
Otherwise, we check its value: if it equals to one of $\{n-5,n-4,\ldots,n-1\}$,
then it is a special symbol since none of these five integers belongs to this bucket.

There is one more subtlety.
In our interior counter trick, there is a counter for each bucket
to count how many suffixes have been placed into this bucket.
Thus there are three types of entries in $\SA$: 1) indices of $T$; 2) these five integers (as special symbols);
3) counters for the buckets.
We want to point out that the position of each counters is always fixed to be adjacent to the special symbol.
Besides, the value of any counter is less than its bucket size minus 2
since the counter is deleted when we place the last two indices into its bucket (see e.g. the figure in Section \ref{sec:roint2}).
Moreover, it is not hard to verify that the counter can only conflict with the special symbol $E$
(which happens when we first insert an index to a bucket).
Thus we choose to use $n-1$ to denote the special symbol $E$.
There is no conflict since the counter is always less than $n-2$.

\section{Restoring the original string $T$}
\label{app:rec}

In this appendix, we show that we can restore the string $T$ in our first algorithm
for  the integer alphabets $\{1,2,\ldots,\Sigma\}$.  First, we can see that upon the termination of our algorithm, suffix array $\SA$ contains the indices of all suffixes of $T$ which are in lexicographical order. Note that if we do not modify $T$, we will have the following observation.
\begin{observation}\label{ob:recover}
For each suffix $\Suf(\SA[i])$ in $\SA$, let $b_i$ denote its bucket character (i.e., the first common character).
We have that $T[\SA[i]]=b_i$.
\end{observation}
The key point to recover $T$ is that we need to maintain the equality relationship of the characters of $T$. So if we modify $T$ to $T'$ under this condition such that $T'[i]=T'[j]$ (or $T'[i]\neq T'[j]$, resp.) if and only if $T[i]=T[j]$ (or $T[i]\neq T[j]$, resp.). Then, we can recover $T$ from $\SA$ and $T'$ using linear time (scan $\SA$ once) and $O(1)$ workspace from above Observation \ref{ob:recover}.
Now, we need to modify the first renaming step in our algorithm to rename each character $T[i]$ to be its bucket tail (note that this modification maintains the equality relationship). This change only leads to some minor change in the details in the later induced sorting step. In the induced sorting step, since we let all $T[i]$ points to its bucket tail, the induced sort LMS-suffixes or S-suffixes are the same as before.
The only thing we need to explain in the induced sorting step is that we induced sort L-suffixes from the sorted LMS-suffixes since there are no pointers which point to the bucket head (see Step 1 of Section \ref{secint:sortt} which sort all L-suffixes from the sorted LMS-suffixes using induced sorting). However, we can fix this step using our interior counter trick we developed. Now, we describe the details. We distinguish the buckets in $\SA$ into two types: one type does not contain LMS-suffixes, the other contains LMS-suffixes. These two types are easy to identify since the LMS-suffixes have already been sorted in the tail of their corresponding buckets in $\SA$.

\topic{Type 1. The buckets do not contain LMS-suffixes}
In this type, we initialize the bucket in the following steps. Scan $T$ from right to left. For every $T[i]$ which is L-type and its bucket is this type, do the following:
\begin{enumerate}[(1)]
    \item If $\SA[T[i]]= \Empty$, let $\SA[T[i]]= \Uniquei$ (unique L-type character in this bucket).
    \item If $\SA[T[i]] = \Uniquei$, let $\SA[T[i]] = \NUniquei$ and $\SA[T[i]-1] = 2$ (number of L-type characters in this bucket is $2$).
    \item If $\SA[T[i]] = \NUniquei$, increase $\SA[T[i]-1]$ by one. ($\SA[T[i]-1]$ denote the number of L-type characters in this bucket)
\end{enumerate}
After this initialization, the head of this type bucket can be indicated by $\SA[t]$ and $\SA[t-1]$, where $t$ is its bucket tail.

\topic{Type 2. The buckets contain LMS-suffixes}
In this type, we initialize the bucket in the following steps. Scan $T$ from right to left. For every $T[i]$ which is L-type and its bucket is this type, do the following:
\begin{enumerate}[(1)]
    \item If $\SA[T[i]]$ is an index, shift these LMS-suffixes (which are sorted in this bucket tail) to left by one position and let $\SA[T[i]]= \Uniqueii$ (unique L-type character in this bucket).
    \item If $\SA[T[i]] = \Uniqueii$, shift these LMS-suffixes (which have been shifted by one position) to left by one position again and let $\SA[T[i]-1] = 2$ (number of L-type characters in this bucket is $2$).
    \item If $\SA[T[i]] = \NUniqueii$, increase $\SA[T[i]-1]$ by one. ($\SA[T[i]-1]$ denote the number of L-type characters in this bucket)
\end{enumerate}
After this initialization, the head of this type bucket can be indicated by $\SA[t]$ and $\SA[t-1]$ too, where $t$ is its bucket tail.

Now, all L-suffixes can be sorted using induced sort like before, but their indices are not in their final positions in $\SA$. We need scan $T$ once more from right to left to compute the number of suffixes in each bucket, then shift these sorted L-suffixes to their bucket head (it is not hard to see that this shift step can be done in linear time). Now, all L-suffixes are placed in their final positions in $\SA$. Then using induced sort as before we can sort all S-suffixes as well. In conclusion, we have the following lemma.
\begin{lemma}
The original string $T$ can be restored using linear time and $O(1)$ workspace.
\end{lemma}

\end{document}